\documentclass[lettersize,journal]{IEEEtran}
\IEEEoverridecommandlockouts
\usepackage{cite}
\usepackage{amsmath,amssymb,amsfonts}
\usepackage{algorithmic}
\usepackage{graphicx}
\usepackage{textcomp}
\usepackage{xcolor}
\def\BibTeX{{\rm B\kern-.05em{\sc i\kern-.025em b}\kern-.08em
    T\kern-.1667em\lower.7ex\hbox{E}\kern-.125emX}}

\usepackage{algorithm}
\usepackage{array}
\usepackage[caption=false,font=normalsize,labelfont=sf,textfont=sf]{subfig}
\usepackage{textcomp}
\usepackage{stfloats}
\usepackage{url}
\usepackage{verbatim}
\usepackage{cite}
\usepackage{xcolor}
\usepackage{amsthm}
\usepackage{ascmac}
\usepackage{enumitem}
\usepackage{tabularx}
\usepackage{multirow}
\usepackage{etoolbox}
\usepackage{lipsum}
\usepackage{colortbl}
\usepackage{hyperref}
\usepackage{orcidlink}
\usepackage{subcaption}
\usepackage{etoolbox}
\usepackage{orcidlink}
\usepackage[font=footnotesize]{caption}
\usepackage{flushend}

\newtheorem{Theorem}{Theorem}
\newtheorem{Definition}{Definition}
\newtheorem{Proposition}{Proposition}
\newtheorem{Remark}{Remark}
\newtheorem{Lemma}{Lemma}

\newtheorem{Corollary}{Corollary}

\DeclareMathAlphabet\mathbfcal{OMS}{cmsy}{b}{n}

\newcommand{\argmin}{\mathop{\rm argmin}\limits}
\newcommand{\argmax}{\mathop{\rm argmax}\limits}

\let\mybibitem\bibitem
\renewcommand{\bibitem}[1]{%
    \ifboolexpr{
    test {\ifstrequal{#1}{ac1993}}
  }
    {\color{black}\mybibitem{#1}}
    {\color{black}\mybibitem{#1}}%
}

\def\A{\mathbf{A}}

\def\N{\mathbf{N}}

\def\H{\mathbf{H}}
\def\tH{\Tilde{\mathbf{H}}}

\def\RG{\mathcal{R}_{\rm G}}
\def\RC{\mathcal{R}_{\rm C}}

\def\tX{\Tilde{\mathbf{X}}}
\def\tx{\Tilde{\mathbf{x}}}
\def\tZ{\Tilde{\mathbf{Z}}}
\def\tz{\Tilde{\mathbf{z}}}
\def\Y{\mathbf{Y}}
\def\y{\mathbf{y}}
\def\Z{\mathbf{Z}}

\def\bSigma{{\bf \Sigma}}

\hypersetup{
    colorlinks=true,  
    linkcolor= blue,  
    citecolor=blue,  
    urlcolor=black  
}

\begin{document}

\title{Secret-Key Generation from Private Identifiers \\ under Channel Uncertainty}

\author{Vamoua~Yachongka$^{\orcidlink{0000-0001-7862-0005}}$,~\IEEEmembership{Member,~IEEE,} and R\'emi A. Chou$^{\orcidlink{0000-0003-4431-3175}}$,~\IEEEmembership{Member,~IEEE,}
\thanks{V.\ Yachongka and R.\ Chou are with the Department of Computer Science and Engineering, The University of Texas at Arlington, Arlington, TX 76019 USA. Corresponding author: Vamoua Yachongka (e-mail: va.yachongka@ieee.org).}
\thanks{A shorter version of this paper was presented at the 2024 IEEE Information Theory Workshop (ITW 2024) \cite{va-remi-itw}. This work was supported in part by NSF grant CCF-2425371.}
}
\markboth{IEEE TRANSACTIONS ON INFORMATION FORENSICS AND SECURITY,}%
{Shell \MakeLowercase{\textit{et al.}}: A Sample Article Using IEEEtran.cls for IEEE Journals}

\maketitle

\begin{abstract}
This study investigates secret-key generation for device authentication using physical identifiers, such as responses from physical unclonable functions (PUFs). The system includes two legitimate terminals (encoder and decoder) and an eavesdropper (Eve), each with access to different measurements of the  identifier. From the device identifier, the encoder generates a secret key, which is securely stored in a private database, along with helper data that is saved in a public database accessible by the decoder for key reconstruction. Eve, who also has access to the public database, may use both her own measurements and the helper data to attempt to estimate the secret key and identifier. Our setup focuses on authentication scenarios where channel statistics are uncertain, with the involved parties employing multiple antennas to enhance signal reception. Our contributions include deriving inner and outer bounds on the optimal trade-off among secret-key, storage, and privacy-leakage rates for general discrete sources, and showing that these bounds are tight for Gaussian sources.
\end{abstract}

\begin{IEEEkeywords}
Capacity region, compound channels, multiple outputs, key generation, privacy leakage, PUFs.
\end{IEEEkeywords}

\section{Introduction}
\IEEEPARstart{T}{he} Internet of Things (IoT) is a rapidly growing technology that enables numerous sensors and small-chip devices to interact and exchange information over the internet. However, ensuring security and privacy in IoT communications presents significant challenges compared to conventional networks due to the diverse range of applications and resource constraints of these devices \cite{Mohamad-2022}. To help address these difficulties, recent efforts have focused  on developing security protocols at the physical layer for authenticating devices.

Secret-key generation using physical identifiers, such as responses from physical unclonable functions (PUFs), is a promising protocol for device authentication because it offers several advantages, including simple designs, low costs, and eliminating the need to save the secret key on the device \cite{gunlue2020}. A PUF is defined as a physical function that for a given input (challenge), provides an output (response) that serves as a unique identifier for each device \cite{pappu2001,puf2002}.
Some examples of PUFs are static random access memory (SRAM) PUFs and ring-oscillator (RO) PUFs.

\begin{figure}[!t]
\centering
   \vspace{-1mm}
\includegraphics[scale=0.64]{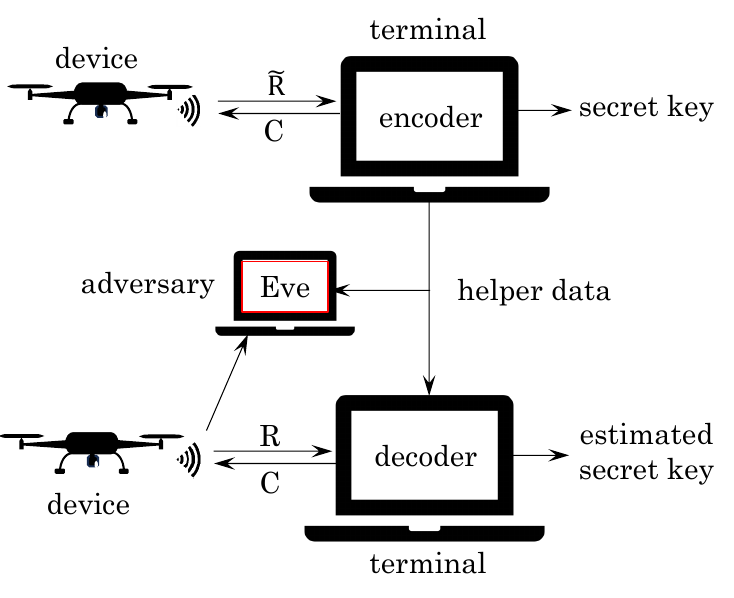}
   \vspace{-2mm}
   \caption{An authentication scheme based on secret-key generation with PUFs. The eavesdropper (Eve) is a passive adversary who is interested in learning the secret key and the source identifier, but does not interfere with the communication mechanism of the system.}
   \vspace{-3mm}
   \label{intro-fig}
\end{figure}

Secret-key authentication using PUFs is illustrated in Fig.~\ref{intro-fig}~\cite{suh-2007}\footnote{\cite{suh-2007} does not consider the presence of Eve, but is included in the figure to facilitate understanding of our system model in Section \ref{sect2}.} and consists of an enrollment phase and an authentication phase. During the enrollment phase, the terminal (encoder) challenges the device, i.e., the PUF embedded in the device, with a challenge C and gets a response $\widetilde{\rm R}$, from which the encoder generates a secret key and helper data. The secret key is securely stored in a private database, while the helper data are saved in a public database, which can be accessed by both the decoder and Eve. In the authentication phase, the terminal (decoder) challenges the device by sending the same challenge C, which produces a different response R due to noise effects. The decoder reconstructs the secret key based on both the response and helper data from the public database and then compares it with the one saved in the private database. If they match, the device is successfully authenticated; otherwise, the authentication fails.

In subsequent discussions, the response of a PUF unaffected by noise is referred to as the {\em source identifier}. The responses of a PUF observed at the terminals and Eve through communication channels  are called the {\em observed identifiers}. It is worth noting that the source identifiers are assumed to be fixed and thus the secret-key generation model considered in this paper corresponds to the source-type model, unlike the channel-type models, where the source distribution can be controlled \cite{ac1993}.

\subsection{Motivations}\label{secm}
We study the capacity region of secret-key generation from source identifiers for a setup involving compound authentication channels with multiple outputs. The motivation for considering compound-channel settings is to capture a situation in which device authentication takes place in environments where the channel statistics may not be perfectly known. This contrasts with most previous studies, which assume that the encoder and decoder have perfect knowledge of the source and channel statistics of the systems.

For example, as shown in Fig. \ref{intro-fig}, consider a situation where the decoder needs to authenticate a flying drone. As the channel state information (CSI) of the channel from the drone to the decoder may fluctuate, it makes it difficult for the decoder to obtain the exact CSI. Additionally, Eve is unlikely to share her CSI with the legitimate terminals. Thus, compound channels are used to model the channels to the decoder and Eve. In this setting, the encoder and decoder do not possess precise CSI of the relevant channels but are aware that these channels belong to certain predefined sets.

Additionally, we consider multiple outputs for the channels to the encoder, decoder, and Eve to capture the circumstance where these parties may deploy multiple antennas to enhance signal reception. Note that having more antennas can increase the correlation between Eve's observation and the source identifier, giving her an advantage in learning the secret key and the source identifier. Hence, in this setting,  we want to quantify the potential leakage to Eve from both security and privacy perspectives.

Finally, in practice, certain types of PUFs produce continuous-value identifiers. For example, the source of RO PUFs can be modeled as a Gaussian distribution \cite{gebali-22}. Additionally, a number of communication channels are sometimes approximated as additive white Gaussian noise (AWGN) channels. This motivates us to study setups with Gaussian sources and AWGN channels.

\subsection{Related Work} \label{secrw}

Secret-key generation using PUFs\footnote{PUF and biometric identifiers share similar characteristics, and thus, the theoretical results developed for one can be applied to the other as well \cite{gunlue2020}.} has been studied from information-theoretic perspectives in \cite{itw3,lhp}. Later, several extensions of this model were found in \cite[Ch. 4]{iw-book} for Gaussian sources, \cite{chou2019} for separated and combined enrollments, and \cite{kuster2019} for multiple rounds of enrollments and authentications. Limited storage rate was introduced to the model in \cite{KY}. Furthermore, the fundamental limits among secret-key, storage, and privacy-leakage rates when Eve also has a correlated sequence of the source were characterized in \cite{kc2015} for discrete sources and \cite{vyo2022-isit} for Gaussian sources. This model is similar to the key-agreement problem with forward communication only studied in \cite{cn2000}, \cite{wataoha-tifs}, \cite{chou2015}, but an additional privacy constraint is imposed in the problem formulation to limit information leakage on the source identifier.

More recent works have considered a setup that incorporates a noisy channel in the enrollment phase \cite{gunlu2018,gksc2018,vamoua-tifs}. The channel is modeled to account for the noise introduced to the source identifier during the enrollment process, providing a more general framework, as signals generated by a PUF are inherently affected by noise. Further progress in this setting has been investigated in \cite{itw-tit-2015,kitti-2016,vy2020,zhou2022}, addressing user identification.

Secret-key generation with PUFs for compound sources has been studied in \cite{tavangaran2017,Grigorescu2017} for  the generated-secret (GS) model and the chosen-secret (CS) model. In the GS model, the secret key is generated using the observed identifier at the encoder. In contrast, the CS model assumes that the secret key is independently and uniformly chosen in advance. Relevant applications of the GS model include field-programmable gate array (FPGA) based key generation with PUFs \cite{colombier-2017,ANA2021} and that of the CS model can be seen in key-binding biometric authentication \cite{Jain-2008} and fuzzy commitment schemes \cite{Ignatenko-2010,gunlu-icc}. Some extensions on this setting are explored in \cite{vu-2021,zhou-allerton} to incorporate user identification. Similar problems can be found in \cite{Remi-2013-GlobalSIP,tavangaran2017-tifs,tavangaran-2018,sultana-2021} for key generation where the privacy constraint is not imposed and in \cite{liang-compound,Bjelakovic-2013,campello-2020,remi-2022-tifs} for compound wiretap channels. We note that the works \cite{liang-compound,Bjelakovic-2013,campello-2020,remi-2022-tifs} focus on compound channels under the channel-type model, whereas our work addresses compound structure in the source-type model.

\subsection{Main Challenges and Contributions}

We begin by explaining challenges of proving the achievability part for general discrete sources. In \cite{tavangaran2017-tifs}, key generation for compound sources without the privacy constraint is investigated. While \cite{tavangaran2017-tifs} derives a single-letter inner bound for discrete sources  and a single-letter outer bound for degraded sources, we establish single-letter  inner and outer bounds for  discrete sources and also characterize the capacity region for Gaussian sources, which require different approaches from the discrete case. The key differences for inner-bound derivations between the work \cite{tavangaran2017-tifs} and ours are twofold.

First, the techniques used for analyzing the secret-key uniformity and secrecy-leakage constraints are distinct. In \cite[Th.~1]{tavangaran2017-tifs}, the secret key is derived from the shared randomness between encoder and decoder, and the analyses of the two constraints rely on extending the method proposed in \cite{CK-book} for non-compound sources. Our approach, in contrast, aligns with \cite{gksc2018}, where the secret key is generated through index mapping, and the analyses of the constraints build upon the technique used in \cite{liang-compound} for analyzing the secrecy constraint in compound wiretap channels.

Second, the privacy-leakage constraint is not considered in \cite{tavangaran2017-tifs}. In our problem formulation, as in \cite{gksc2018,vamoua-tifs}, this constraint is imposed and quantified by the mutual information between the source identifier and the helper data, conditioned on Eve's observation. Its analysis is not straightforward because the helper data does not have an independent and identically distributed (i.i.d.) structure: although the encoder observes an i.i.d. identifier sequence, it generates the helper data based on the entire block rather than on individual symbols.

In the converse part, for a given channel state, the proofs of the GS and CS models mirror the ones in \cite[Th.~3~and~4]{gksc2018} with a proper modification for the privacy-leakage analysis as the definition is distinct. These results are then generalized to the compound-channel settings by taking the intersection over all possible channel states to establish the outer bounds. As a result, the inner region first involves an optimization carried over the distributions of auxiliary random variables, and then a minimization of the index pair for channel states. In contrast, the order of these two operations is reversed in the outer region. This leads to a gap between the inner and outer bounds, similar to the conclusion drawn in \cite{liang-compound} for compound wiretap channels.

However, we show that, for a noiseless enrollment channel, our inner and outer regions coincide for Gaussian sources, providing a complete capacity characterization. The main challenging aspect arises in proving the converse part. Given the multiple-antenna settings at the legitimate terminals and Eve, the vector-form observations are not stochastically degraded in general \cite{weingarten-2006}. We use sufficient statistics \cite{cover} to convert the vector problem into a scalar one. However, after this conversion, showing that all constraints of the original problem definition are preserved is challenging, and it is unclear whether the same expressions of the outer bounds for general discrete sources also hold for the scalar variables. Therefore, we cannot directly apply the technique in \cite[Appx. B]{wataoha2010} to eliminate the second auxiliary random variable. In this paper, we instead derive new single-letter expressions of outer bounds for the Gaussian case using  scalar random variables.

Another difficulty arises in proving the converse for the parametric expression of Gaussian sources. In the analysis of the model without side information at Eve \cite[Appx. D]{iw-book}, \cite{vy2}, the conditional entropy power inequality (EPI) plays an important role. However, the EPI is insufficient to prove the converse for all possible values of the optimization parameter in our problem. To overcome this issue, we adopt a distinct method introduced in \cite[Sect. IV-C]{ekrem-2014}, using Fisher information. This approach enables us to derive the outer region that coincides with the inner one for any value of the optimization parameter.

Our main contributions are summarized as follows:
\begin{itemize}
    \item We derive inner and outer bounds on the capacity regions of secret-key, storage, and privacy-leakage rates of the GS and CS models for discrete sources.
    \item We provide complete characterizations of the capacity regions of the GS and CS models for Gaussian sources by demonstrating the existence of a saddle point at which the inner and outer bounds coincide.
    \item {We conduct numerical calculations for the Gaussian case to illustrate how the change of the number of antennas at the decoder and Eve affects the secret-key and storage rates. The results show that increasing the number of antennas at the decoder leads to a higher secret-key rate, while increasing antennas at Eve reduces the secret-key rate. Nevertheless, {even if Eve has more antennas}, a positive secret-key rate is still achievable as long as the worst channel power gain at the decoder is greater than the best channel power gain at Eve. Moreover, we compare the secret-key and privacy-leakage rates between the GS and CS models under the same values of storage rates. The results reveal that in the low storage-rate regime, the GS model outperforms {the CS model} in terms of secret-key rate, whereas the CS model provides better privacy-leakage performance. In the high storage-rate regime, the CS model is {better suited as it can achieve the same secret-key rate as the GS model but with lower privacy leakage}}.
\end{itemize}

 Our results recover, as special cases,  results derived in previous works. For discrete sources, when only the channel to the decoder is compound and all channels have a single output, the inner and outer bounds match the preliminary result given in \cite[Props. 1 and 2]{va-remi-itw}. Additionally, the inner and outer bounds are tight for single-output and non-compound channels, and recover \cite[Th. 3 and 4]{gksc2018} without action cost. For Gaussian sources, as detailed in Section~\ref{continuous-sect}, our results recover as special cases the capacity regions derived in \cite[Ch. 4]{iw-book} and \cite{vamoua-tifs} for single-output and non-compound channels.

\subsection{{Modeling Assumptions}}
{In general, PUF responses from devices are correlated. However, techniques such as transform coding-based algorithms \cite{gunlue2020} and principal component analysis \cite{kelkboom-2010} can be applied to convert these responses into a sequence with almost independent symbols. Therefore, we assume that each symbol in the source and observed identifier sequences is i.i.d. generated. Additionally, we assume that the database that stores helper data is public, e.g., in the cloud, and accessible to both the decoder and Eve. These modeling assumptions are consistent with those used in prior works \cite{itw3,lhp,gunlu2018,gksc2018,vamoua-tifs}}.

\subsection{Notation and Paper Organization}
$\mathbb{R}_+$ is the set of non-negative real numbers. For any $a,b \in \mathbb{R}$, define $[a:b] \triangleq [\lfloor a \rfloor,\lceil b \rceil] \cap \mathbb{N}$. Italic uppercase $X$ and lowercase $x$ denote a random variable and its realization, respectively. Boldface letters $\bf X$ and $\bf x$ represent a collection of random variables and its realization. $X^n$ denotes the vector $ (X_1,{\dots},X_n)$ and $X_t$ represents the $t$-th element in the  vector.  $X_k^t$ stands for a partial sequence $(X_k,\dots,X_t)$ for any $[k:t] \subseteq [1:n]$.  $\sigma^2_X$ and $\bSigma_{\Y}$ denote the variance of $X$ and the covariance matrix of ${\bf Y}$. $\mathcal{N}(0,\sigma^2)$ denotes the Gaussian distribution with zero mean and variance $\sigma^2$.
$\mathcal{T}^n_{\delta}(X)$ denotes the set of $\delta$-strongly typical sequences according to $P_X$ \cite{GK} and the random variable inside the parentheses is omitted, e.g., $\mathcal{T}^n_{\delta}$, when it is clear from the context. Additionally, the set of conditionally $\delta$-typical sequences is denoted as $\mathcal{T}^n_{\delta}(XY|z^n) \triangleq \{(x^n,y^n): (x^n,y^n,z^n) \in \mathcal{T}^n_{\delta}\}$ {for a given $z^n \in \mathcal{Z}^n$}.

The remainder of the paper is organized as follows. In Section~\ref{sect2}, we state the problem definitions for the GS and CS models. We present our main results in Section \ref{sect3}. Proofs of our main results are available in the appendices. Finally, we provide concluding remarks and some future directions in Section~\ref{sect4}.

\section{Problem Statement} \label{sect2}
The source identifier $X^n$ is i.i.d. according to $P_X$. The terminals do not have direct access to this identifier but can only observe its noisy versions.
The encoder, decoder, and Eve are equipped with $(\Omega_{\Tilde{X}},\Omega_Y,\Omega_Z) \in \mathbb{N}^3$ receiver antennas, respectively. Furthermore, there are {$|\mathcal{K}|$} possible states for the channel to the decoder $P_{\Y_k|X}$ with $k \in \mathcal{K}$, and {$|\mathcal{L}|$} possible states for the channel to Eve $P_{\Z_l|X}$ with $l \in \mathcal{L}$ in the authentication phase.\footnote{Considering multiple states for the enrollment channel is unnecessary since its state could be estimated at the encoder and shared with the decoder through the helper data with a negligible cost.}  When the channels are in State $(k,l)$, the setting is depicted in Fig.~\ref{system-model}. The vector-form random variables $\tX^n \triangleq [\Tilde{X}^n_1,\Tilde{X}^n_2,\cdots,\Tilde{X}^n_{\Omega_{\Tilde{X}}}]^\intercal$, $\Y^n_k \triangleq [Y^n_{k1},\cdots,Y^n_{k\Omega_Y}]^\intercal$, and $\Z^n_l \triangleq [Z^n_{l1},\cdots,Z^n_{l\Omega_Z}]^\intercal$ denote the outputs of the source identifier $X^n$ via the channel to the encoder, $(\mathcal{X},P_{\tX|X}, \Tilde{\mathbfcal{X}})$, and the channels to the decoder and Eve $(\mathcal{X},P_{\Y_k\Z_l|X}, \mathbfcal{Y}_k\times \mathbfcal{Z}_l)$, respectively. The joint distribution of the system is
\begin{align}
    P_{\Tilde{\bf X}^nX^n{\bf Y}^n_k{\bf Z}^n_l}
    &\triangleq \textstyle \prod_{t=1}^n P_{\Tilde{\bf X}_t|X_t}\cdot P_{X_t} \cdot P_{\Y_{k,t}\Z_{l,t}|X_t}. \label{jointd}
\end{align}

Secret-key generation strategies  are formally defined below. Let $\mathcal{S} \triangleq [1:2^{nR_S}]$ and $\mathcal{J} \triangleq [1:2^{nR_J}]$.

\begin{figure}[!t]
\centering
\includegraphics[scale=0.77]{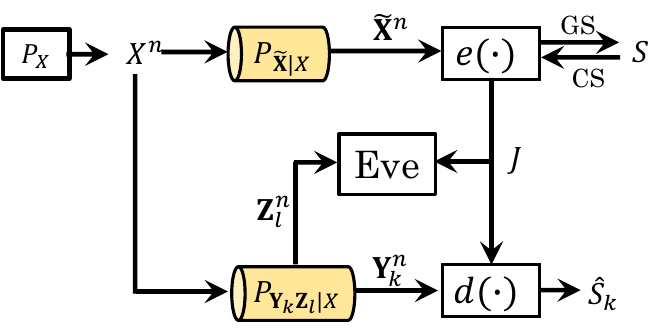}
\caption{{Illustration of the system model in State $(k,l)$}.}
\label{system-model}
\end{figure}

\begin{Definition}[GS model] \label{enc-dec-gs}
 For the GS model, a $(2^{nR_S},n,R_J,R_L)$ secret-key generation strategy  consists of:
\begin{itemize}
    \item Encoding mapping $e:\Tilde{\mathbfcal{X}}^n \rightarrow \mathcal{J}\times \mathcal{S}$;
    \item Decoding mapping $d:\mathbfcal{Y}_k^n\times\mathcal{J} \rightarrow \mathcal{S}$;
\end{itemize}
and operates as follows:
\begin{itemize}
    \item The encoder observes  $\tX^n$ and generates the helper data $J \in \mathcal{J}$ and a secret key $S \in \mathcal{S}$, as $(J,S) \triangleq e(\tX^n)$;
    \item The helper data $J$ is stored in a public database, accessible to anyone;
    \item From $\Y^n_k$ and  $J$, retrieved from the public database, the decoder estimates $S$   as $\hat{S}_k \triangleq d(\Y^n_k,J)$.
\end{itemize}
\end{Definition}
\begin{Definition}[CS model] \label{enc-dec-cs}
For the CS model, a $(2^{nR_S},n,R_J,R_L)$ secret-key generation strategy consists of:
\begin{itemize}
    \item A secret key {$S$}, chosen uniformly at random in $\mathcal{S}$ and independently of $(\tX^n,X^n,\Y^n_k,\Z^n_l)$;
    \item Encoding mapping $e:\Tilde{\mathbfcal{X}}^n{}\times \mathcal{S} \rightarrow \mathcal{J}$;
    \item Decoding mapping $d:\mathbfcal{Y}^n_k\times\mathcal{J} \rightarrow \mathcal{S}$;
\end{itemize}
and operates as follows:
\begin{itemize}
    \item  From $\tX^n$ and  $S$, the encoder generates  $J \triangleq e(\tX^n,S)$;
    \item The helper data $J$ is saved in a public database, accessible to anyone;
    \item From $\Y^n$ and $J$, the decoder  estimates $S$  as $\hat{S}_k \triangleq d(\Y^n_k,J)$.
\end{itemize}
\end{Definition}

In the following, we write $(\displaystyle \max_{k \in \mathcal{K}},\displaystyle \min_{k \in \mathcal{K}})$ and $(\displaystyle \max_{l \in \mathcal{L}},\displaystyle \min_{l \in \mathcal{L}})$ as $(\displaystyle \max_{k},\displaystyle \min_{k})$ and $(\displaystyle \max_{l},\displaystyle \min_{l})$, respectively, for simplicity.

\begin{Definition}[GS model] \label{def1}
A tuple of secret-key, storage, and privacy-leakage rates $(R_S,R_J,R_L)\in \mathbb{R}^3_+$ is achievable for the GS model if, for sufficiently small $\delta > 0$ and large enough $n$, there exist pairs of encoders and decoders satisfying
\begin{align}
\max_{k}\mathbb{P}\{\widehat{S}_k \neq S\} &\leq  \delta,\label{errorp} \\
H(S) + n\delta &\ge \log |\mathcal{S}| \ge  n(R_S - \delta), \label{secretk} \\
\log{|\mathcal{J}|} &\leq n(R_J + \delta), \label{storage} \\
\max_{l}I(S;J,\Z^n_l) &\leq n\delta, \label{secrecy} \\
\max_{l}I(X^n;J|\Z^n_l) &\leq n(R_L + \delta). \label{privacy}
\end{align}
$\RG$ is defined as the closure of the set of all achievable rate tuples for the GS model, and it is called the capacity region.
\end{Definition}
\begin{Definition}[CS model] \label{def2}
A tuple of secret-key, storage, and privacy-leakage rates $(R_S,R_J,R_L)\in \mathbb{R}^3_+$ is achievable for the CS model if, for sufficiently small $\delta > 0$ and large enough $n$, there exist pairs of encoders and decoders that satisfy all the conditions \eqref{errorp}--\eqref{privacy} with replacing \eqref{secretk} by
$
\log |\mathcal{S}| \ge  n(R_S - \delta).
$
Let $\RC$ be the capacity region of the CS model.
\end{Definition}

In Definition \ref{def1}, \eqref{errorp} denotes the reliability constraint, \eqref{secretk} is the uniformity requirement of the generated secret key, \eqref{storage} is the storage rate constraint, \eqref{secrecy} is the secrecy-leakage constraint, evaluating the information about the secret key leaked to Eve, and \eqref{privacy} is the privacy-leakage constraint, quantifying the amount of information leaked to Eve regarding the source identifier via the helper data given Eve's side information.

\section{Main Results} \label{sect3}
This section presents the inner and outer bounds for the GS and CS models with discrete sources, followed by the tight bounds for Gaussian sources and numerical results for both models.

\subsection{Discrete Sources} \label{sec:resd}

\begin{Proposition}[Inner bounds] \label{inner-bound}
We have
\begin{align}
    \RG
    &\supseteq \bigcup_{P_{V|U},P_{U|\tX}}\Big\{(R_S,R_J,R_L)\in \mathbb{R}^3_+: \nonumber \\
    R_S &\le \min_kI(\Y_k;U|V) - \max_lI(\Z_l;U|V), \nonumber \\
    R_J &\ge \max_kI(\tX;U|V,\Y_k) + \max_kI(\tX;V|\Y_k), \nonumber \\
    R_L &\ge \max_kI(\tX;U|V,\Y_k) + \max_kI(\tX;V|\Y_k) \nonumber \\
    &~-I(\tX;U|X)+ \min_kI(\Y_k;V)-\min_lI(\Z_l;V) \Big\},
    \label{inner-gs} \displaybreak[0]\\
    \RC
    &\supseteq \bigcup_{P_{V|U},P_{U|\tX}}\Big\{(R_S,R_J,R_L)\in \mathbb{R}^3_+: \nonumber \\
    R_S &\le \min_kI(\Y_k;U|V) - \max_lI(\Z_l;U|V), \nonumber \\
    R_J &\ge \max_kI(\tX;U|V,\Y_k) + \max_kI(\tX;V|\Y_k) \nonumber \\
    &~+\min_kI(\Y_k;U|V) - \max_lI(\Z_l;U|V), \nonumber \\
    R_L &\ge \max_kI(\tX;U|V,\Y_k) + \max_kI(\tX;V|\Y_k) \nonumber \\
    &~-I(\tX;U|X)+ \min_kI(\Y_k;V)-\min_lI(\Z_l;V) \Big\},
    \label{inner-cs}
\end{align}
where auxiliary random variables $V$ and $U$ satisfy the Markov chain $V-U-\tX-X-(\Y_k,\Z_l)$ for all $k \in \mathcal{K}$ and $l \in \mathcal{L}$, and $|\mathcal{V}| \le |\Tilde{\mathbfcal{X}}| + 6$ and $|\mathcal{U}| \le (|\Tilde{\mathbfcal{X}}| + 2(|\mathcal{K}|+|\mathcal{L}|)+1)(|\Tilde{\mathbfcal{X}}| + |\mathcal{K}|+|\mathcal{L}|) +1)$.
\end{Proposition}
\begin{proof}
The proof is available in Appendix \ref{proof-inner}, where the random codebook is constructed based on two layered random coding techniques. The first layer consists of the auxiliary sequences $V^n$ generated by $P_V$ and the second layer consists of the auxiliary sequences $U^n$ associated with $P_{U|V}$. The main challenge in the proof is to ensure that the secret-key uniformity \eqref{secretk} and the secrecy-leakage constraint \eqref{secrecy} are satisfied for all possible receiver-eavesdropper states. To prove these constraints, the key idea is to introduce a random variable $\tZ^n$ that jointly satisfies the equality $I(\tZ;U|V) = \max_{l}I(\Z_l;U|V)$ and {the Markov chain $V-U-\tX-\tZ$}. The random variable $\tZ^n$ plays a {central role in analyzing} the two constraints. This technique is not seen in the existing works \cite{itw3,gksc2018,vamoua-tifs} that study the secret-key generation with PUFs without compound channels.
\end{proof}

In Proposition \ref{inner-bound}, how each term in the constraints defining the regions $\RG$ and $\RC$ arises can be explained as follows. {We begin with the region $\RG$. In the secret-key rate constraint, the term $\min_k I(\Y_k; U | V)$ represents the minimum rate required for reliably estimating the sequence $U^n$ across all indices $k$, which in turn enables reliable reconstruction of the secret key since the key is extracted from $U^n$. On the other hand, the term $\max_lI(\Z_l;U|V)$ is the maximum rate at which Eve can gain information about $U^n$ over all indices $l$. Therefore, the achievable secret-key rate is given by the difference $\min_kI(\Y_k;U|V) - \max_lI(\Z_l;U|V)$, similar to the one derived in \cite[Th.~1]{tavangaran2017-tifs} for compound sources}. The terms $\max_kI(\tX;V|\Y_k)$ and $\max_kI(\tX;U|V,\Y_k)$ in the storage-rate constraint represents the rates of the bin indices at the first and second layers, respectively. {In each layer, the maximum rate across all indices $k$ must be shared between the encoder and decoder to ensure reliable reconstruction of the secret key at the decoder}. For the privacy-leakage rate, note that we can expand the mutual information $\frac{1}{n}\max_{l}I(X^n;J|\Z^n_l)$ as $\frac{1}{n}H(J) - \frac{1}{n}H(J|X^n) - \frac{1}{n}\min_{l}I(\Z^n_l;J)$ by using the Markov chain $J-X^n-\Z^n_l$. In the constraint of the privacy-leakage rate, the first and second terms in the right-hand side represent the upper bound on the entropy $\frac{1}{n}H(J)$, the third term represents the upper bound on the conditional entropy $-\frac{1}{n}H(J|X^n)$, and the forth and fifth terms represent the upper bound on the mutual information $-\frac{1}{n}\min_{l}I(\Z^n_l;J)$.

For the region $\RC$, the codebook and coding scheme developed for proving the region $\RG$ are employed as a subsystem to prove the achievability part. One-time pad operation is applied to conceal the chosen secret key in the CS model by adding the secret key generated in the subsystem \cite[Appx.~B-C]{itw3}, which leads to the same achievable secret-key rate. However, the storage rate is different because the masked information must be saved in the public database together with the helper data generated by the subsystem, so that the chosen secret key can be reliably estimated at the decoder. Therefore, the storage rate of the CS model is the sum of the storage rate for the GS model (the subsystem) and the secret-key rate. Moreover, the privacy-leakage rate remains unchanged because the concealed information reveals no extra leakage to {Eve after} applying the one-time pad addition. Similar behaviors {are reflected} in the outer bound derived below in Proposition \ref{discrete-result-outer}.

A special case of the GS model considered in this paper was investigated in \cite {va-remi-itw}.  One can check that, when  $\Omega_{\Tilde{X}}=\Omega_Y=\Omega_Z=|\mathcal{L}|=1$, the region in \eqref{inner-gs} reduces to \cite[Prop.~1]{va-remi-itw}.

\begin{Proposition} [Outer bounds]\label{discrete-result-outer}
 We have
\begin{align}
    \RG
    \subseteq &\bigcap_{k \in \mathcal{K}}\bigcap_{l \in \mathcal{L}}\bigcup_{P_{V|U},P_{U|\tX}}\Big\{(R_S,R_J,R_L)\in \mathbb{R}^3_+:~\nonumber \\
    R_S &\le I(\Y_k;U|V) - I(\Z_l;U|V), \nonumber \\
    R_J &\ge I(\tX;U|\Y_k), \nonumber \\
    R_L &\ge I(X;U|\Y_k) + I(\Y_k;V)-I(\Z_l;V)\Big\},
    \label{descrete-outer-gs} \displaybreak[0]\\
    \RC
    \subseteq &\bigcap_{k \in \mathcal{K}}\bigcap_{l \in \mathcal{L}}\bigcup_{P_{V|U},P_{U|\tX{}}}\Big\{(R_S,R_J,R_L)\in \mathbb{R}^3_+:~\nonumber \\
    R_S &\le I(\Y_k;U|V) - I(\Z_l;U|V), \nonumber \\
    R_J &\ge I(\tX;U|\Y_k) + I(\Y_k;U|V) - I(\Z_l;U|V), \nonumber \\
    R_L &\ge I(X;U|\Y_k) + I(\Y_k;V)-I(\Z_l;V)\Big\},
    \label{descrete-outer-cs}
\end{align}
where $V$ and $U$ satisfy the Markov chain $V-U-\tX-X-(\Y_k,\Z_l)$, and $|\mathcal{V}| \le |\Tilde{\mathbfcal{X}}| + 2(|\mathcal{K}|+|\mathcal{L}|) + 1$, $|\mathcal{U}| \le (|\Tilde{\mathbfcal{X}}| + 2(|\mathcal{K}|+|\mathcal{L}|)+1)(|\Tilde{\mathbfcal{X}}| + |\mathcal{K}|+|\mathcal{L}| + 1)$.
\end{Proposition}
\begin{proof}
The proof is provided in Appendix \ref{proof-outer}. For a fixed state $(k,l)$, the proof is the same as that of \cite[Th.~3~and~4]{gksc2018} without considering the action cost. Therefore, we make use of the result of those theorems to derive Proposition \ref{discrete-result-outer} for the compound channel setting. However, due to the difference in the definition of the privacy-leakage rate, appropriate modifications are required.
\end{proof}

The region in \eqref{descrete-outer-gs} matches \cite[Prop. 2]{va-remi-itw} when $\Omega_{\Tilde{X}}=\Omega_Y=\Omega_Z=|\mathcal{L}|=1$. Moreover, in the non-compound settings, i.e., when $|\mathcal{K}|=1=|\mathcal{L}|$, the bounds in Propositions \ref{inner-bound} and  \ref{discrete-result-outer} match, yielding the following corollary.


\begin{Corollary}[Capacity regions] \label{coro-capacity}
When $|\mathcal{K}|=1=|\mathcal{L}|$, we~have
\begin{align}
    \RG
    = &\bigcup_{P_{V|U},P_{U|\tX}}\Big\{(R_S,R_J,R_L)\in \mathbb{R}^3_+:~\nonumber \\
    R_S &\le I(\Y;U|V) - I(\Z;U|V), \nonumber \\
    R_J &\ge I(\tX;U|\Y), \nonumber \\
    R_L &\ge I(X;U|\Y) + I(\Y;V)-I(\Z;V)\Big\},
    \label{coro-gs} \displaybreak[0] \\
    \RC
    = &\bigcup_{P_{V|U},P_{U|\tX{}}}\Big\{(R_S,R_J,R_L)\in \mathbb{R}^3_+:~\nonumber \\
    R_S &\le I(\Y;U|V) - I(\Z;U|V), \nonumber \\
    R_J &\ge I(\tX;U|\Y) + I(\Y;U|V) - I(\Z;U|V), \nonumber \\
    R_L &\ge I(X;U|\Y) + I(\Y;V)-I(\Z;V)\Big\},
    \label{coro-cs}
\end{align}
where $(U,V)$ satisfy the same conditions as in Proposition~\ref{discrete-result-outer}.
\end{Corollary}
\begin{proof} We only   sketch the proof of \eqref{coro-gs}, as that of  \eqref{coro-cs} follows similarly. When $|\mathcal{K}|=|\mathcal{L}| = 1$ and by dropping the indices $k$ and $l$, the constraints in Proposition \ref{inner-bound} become
\begin{align}
    R_S &\le  I(\Y;U|V) - I(\Z;U|V), \label{rs-coro} \displaybreak[0]\\
    R_J &\ge  I(\tX;U|V,\Y) + I(\tX;V|\Y)\overset{\rm (a)}= I(\tX;U|\Y), \label{rj-coro} \displaybreak[0]\\
    R_L & \ge I(\tX;U|V,\Y) + I(\tX;V|\Y) \nonumber \\
    &~~~-I(\tX;U|X)+ I(\Y;V)-I(\Z;V) \nonumber \displaybreak[0]\\
    &=I(\tX;U|\Y)-I(\tX;U|X) + I(\Y;V)-I(\Z;V) \nonumber \\
    &\overset{\rm (b)}=I(X;U|\Y) + I(\Y;V)-I(\Z;V), \label{rl-coro}
\end{align}
where (a) and (b) hold by the Markov chains $V-U-\tX-\Y$ and $U-\tX-X-\Y$, respectively. As $|\mathcal{K}|=|\mathcal{L}| = 1$, \eqref{rs-coro}--\eqref{rl-coro} match \eqref{descrete-outer-gs} in Proposition \ref{discrete-result-outer}, and thus \eqref{coro-gs} is proved.
\end{proof}
\begin{Remark}
The privacy-leakage rate in \cite[Th.~3 and 4]{gksc2018} without action cost is bounded as
\begin{align}
    R_L &\ge I(X;U,Y) - I(X;Y|V) + I(X;Z|V) \nonumber \\
    &=I(X;U|Y) + I(Y;V)-I(Z;V) + I(X;Z),
\end{align}
where the equality holds by the Markov chain $V-X-(Y,Z)$.
Compared to the privacy-leakage rate in \eqref{coro-gs} and \eqref{coro-cs}, there is an extra term $I(X;Z)$,  because the privacy-leakage rate constraint in \cite[Th. 3 and 4]{gksc2018} is defined as $\frac{1}{n}I(X^n;J,Z^n) = \frac{1}{n}I(X^n;J|Z^n) + I(X;Z)$. Therefore,  \eqref{coro-gs} and \eqref{coro-cs} coincide with \cite[Th. 3 and 4]{gksc2018} (without action cost) if \cite[eq. (5)]{gksc2018} is replaced by \eqref{privacy}.
\end{Remark}

In Propositions \ref{inner-bound} and \ref{discrete-result-outer}, the orders of the optimization (union) over the test channels $P_{V|U}, P_{U|\tX}$ and the minimization (intersection) over the channel states $(k,l)$ are {reversed}. Specifically, Proposition \ref{inner-bound} requires one to choose test channels $P_{V|U}, P_{U|\tX}$ that work simultaneously for all $(k,l)$ pairs, whereas Proposition \ref{discrete-result-outer} allows one to choose different test channels 
$P_{V|U}, P_{U|\tX}$ for a channel state pair \((k,l)\), and then only keeps the intersection over what is achievable per channel state pair. Therefore, Proposition \ref{inner-bound} imposes stronger requirements, and as a result, the regions in Proposition \ref{discrete-result-outer} may be potentially larger.


In the next subsection, we demonstrate that the regions in Propositions \ref{inner-bound} and   \ref{discrete-result-outer} match for Gaussian sources.

\subsection{Gaussian Sources} \label{continuous-sect}

In this subsection, we limit our discussion to a special case of setup in Section \ref{sec:resd} where  the enrollment channel is noiseless, i.e., $\tX = X$. We consider $P_{X\Y_k\Z_l}$ the joint distribution of zero-mean Gaussian random variables with a non-singular covariance matrix. Suppose that the source $X \sim \mathcal{N}(0,\sigma^2_X)$, then it suffices to model the channels to the decoder and Eve as follows.
\begin{Lemma} \label{lemma-111} Without loss of generality, one can write
\begin{align}
\mathbf{Y}_k = \H_kX  + \mathbf{N}_{\Y_k},~\mathbf{Z}_l = \tH_lX + \mathbf{N}_{\Z_l}, \label{channel-eq1}
\end{align}
where $\H_k \in \mathbb{R}^{\Omega_Y\times 1}$, $\tH_l \in \mathbb{R}^{\Omega_Z\times 1}$, and $\mathbf{N}_{\Y_k} \sim {\bf \mathcal{N}}({\bf 0}, \mathbf{I}_{\Omega_Y})$, and $\mathbf{N}_{\Z_l} \sim {\bf \mathcal{N}}({\bf 0},\mathbf{I}_{\Omega_Z})$ are independent of $X$. Here, $\mathbf{I}$ denotes the identity matrix.
\end{Lemma}
\begin{proof}
    See Appendix \ref{proof-lemma111}.
\end{proof}

\begin{Remark} \label{lemma-txneqx}
In the case where the enrollment channel is noisy, i.e., $\tX \neq X$, the noise covariance matrices of the involved channels are not positive definite in general. This prevents the use of Cholesky decomposition to normalize them to identity matrices, and thus the channel models described in Lemma~\ref{lemma-111} may no longer be applicable.
\end{Remark}

Note that the single-letter expressions characterized in Propositions \ref{inner-bound} and \ref{discrete-result-outer} can be extended to the channel model  \eqref{channel-eq1}. To derive a closed-form analytical (parametric) expression for Gaussian sources, we directly leverage Proposition~\ref{inner-bound} to show the achievability. In the converse, we transform the problem in \eqref{channel-eq1} into a scalar Gaussian problem using sufficient statistics \cite[Ch. 2]{cover}, which helps avoid the complexity of working with vector random variables. However, after the transformation, it is unclear whether all constraints in Definition~\ref{def1}, particularly \eqref{errorp}, remain preserved under the scalar random variables. As a result, Proposition \ref{discrete-result-outer} may not hold when the vector random variables are replaced with scalar ones. To this end, as shown in the proof of Theorem~\ref{th1-gauss}, we derive new outer regions for the channel model \eqref{channel-eq1} using scalar variables to establish the converse part of Theorem \ref{th1-gauss}.

In the sequel, we define
\begin{align}
    &k^* \in \argmin_{k \in \mathcal{K}}\{\H^\intercal_{k}\H_{k}\},~~~l^* \in \argmax_{l \in \mathcal{L}}\{\tH^\intercal_{l}\tH_{l}\}. \label{kl-star}
\end{align}

To simplify the presentation of the results for Gaussian sources, we define the following rate constraints, where $\alpha \in (0,1]$ serves as a tuning parameter that adjusts the variance of the auxiliary Gaussian random variable. For further details, the reader is referred to \eqref{x-u-phi}.
\begin{align}
R_S &\le\frac{1}{2}\log\left(\frac{(\sigma^2_X\H^\intercal_{k^*}\H_{k^*} + 1)(\alpha\sigma^2_X\tH^\intercal_{l^*}\tH_{l^*} + 1)}{(\alpha\sigma^2_X\H^\intercal_{k^*}\H_{k^*} + 1)(\sigma^2_X\tH^\intercal_{l^*}\tH_{l^*} + 1)}\right), \label{rs-gauss-th1} \\
R_J &\ge \frac{1}{2}\log\left(\frac{\alpha\sigma^2_X\H^\intercal_{k^*}\H_{k^*} + 1}{\alpha(\sigma^2_X\H^\intercal_{k^*}\H_{k^*} + 1)}\right), \label{rj-gauss-th1}\\
R_J &\ge \frac{1}{2}\log\left(\frac{\alpha\sigma^2_X\tH^\intercal_{l^*}\tH_{l^*} + 1}{\alpha(\sigma^2_X\tH^\intercal_{l^*}\tH_{l^*} + 1)}\right), \label{rj-gauss-th1-2} \\
R_L &\ge \frac{1}{2}\log\left(\frac{\alpha\sigma^2_X\H^\intercal_{k^*}\H_{k^*} + 1}{\alpha(\sigma^2_X\H^\intercal_{k^*}\H_{k^*} + 1)}\right). \label{rl-gauss-th1}
\end{align}

\begin{Theorem}[Capacity regions] \label{th1-gauss} If $\H^\intercal_{k^*}\H_{k^*} \ge \tH^\intercal_{l^*}\tH_{l^*}$, then the  capacity regions of the GS and CS models are 
\begin{align}
\RG=\bigcup_{0 < \alpha \le 1}\{(R_S,R_J,R_L)\in \mathbb{R}^3_+:~&{\rm \eqref{rs-gauss-th1},~\eqref{rj-gauss-th1},~and~\eqref{rl-gauss-th1}}\nonumber \\
&~{\rm are~satisfied}\} \label{gs-cp-gauss}, \displaybreak[0]\\
\RC=\bigcup_{0 < \alpha \le 1}\{(R_S,R_J,R_L)\in \mathbb{R}^3_+:~&{\rm \eqref{rs-gauss-th1},~\eqref{rj-gauss-th1-2},~and~\eqref{rl-gauss-th1}}\nonumber \\
    &~{\rm are~satisfied}\} \label{cs-cp-gauss}.
\end{align}
If $\H^\intercal_{k^*}\H_{k^*} < \tH^\intercal_{l^*}\tH_{l^*}$, then   
\begin{align}
\RG = \RC = \{(R_S,R_J,R_L):
    R_S = 0, R_J \ge 0, R_L \ge 0\}.
\end{align}
\end{Theorem}
\begin{proof}
The proof is provided in Appendix \ref{th1-proof} and includes the achievability and converse parts. For the achievability, {we} set the test channel $P_{U|X}$ to be an AWGN channel and then apply Weinstein–Aronszajn Identity \cite[Appx. B]{Pozrikidis2014} to calculate the mutual information with vector random variables. Finally, we use Lemma \ref{lemma-fx} to show that the optimal inner region is achieved when the indices of the channels to the decoder and Eve are $k^*$ and $l^*$, respectively. {For the converse, we begin} by invoking the sufficient statistics \cite{cover} to convert vector variables to scalar ones. Next, we derive a single-letter characterization of the outer bound using these scalar variables, which is then used to determine the parametric expressions for the Gaussian case. The proof employs a technique based on Fisher information, introduced in \cite{ekrem-2014}. In the final step, we again apply Lemma \ref{lemma-fx} to derive the outer region valid for an arbitrary pair $(k,l)$, which is obtained when the decoder and Eve observe the channels indexed by $k^*$ and $l^*$ as well, coinciding with the optimal inner bound.
\end{proof}

In Theorem \ref{th1-gauss}, the condition $\H^\intercal_{k^*}\H_{k^*} \ge \tH^\intercal_{l^*}\tH_{l^*}$ indicates that the channel power gain of the worst link to the decoder is at least as large as that of the best link to Eve. In the single-antenna case, i.e., $|\mathcal{K}|=1=|\mathcal{L}|$, this condition corresponds to physically degraded channels, where the channel to Eve is physically degraded with respect to the channel to the decoder.

Unlike the discrete sources, the inner and outer bounds for the Gaussian sources coincide. This is because, in the outer bound, the variable involved in the optimization is a scalar parameter, and rate constraints are given by logarithmic functions of the optimization parameter, $\alpha$, and the values of channel power gains $\H^\intercal_{k}\H_{k}$ and $\tH^\intercal_{l}\tH_{l}$. These functions are monotonic with respect to the channel power gains for an arbitrary $\alpha$. As a result, the order of intersection and union does not matter and can be swapped, which enable us to take the intersection over channel states $(k,l)$ for each rate constraint and determine the saddle point $(k^*,l^*)$ at which the outer bound matches the inner bound.

As a special case, when $\Omega_Y$, $\Omega_Z$, $|\mathcal{K}|$, and $|\mathcal{L}|$ are all one (let $\H=h$ and $\tH=\Tilde{h}$), the AWGN channels to the decoder and Eve reduce to $Y = hX + N_Y$ and $Z = \Tilde{h}X + N_Z$, respectively, with $N_Y \sim \mathcal{N}(0,1)$ and $N_Z \sim \mathcal{N}(0,1)$. In this case, using the correlation coefficients of $(X,Y)$, $\rho^2_{XY} = \sigma^2_Xh^2/(\sigma^2_Xh^2 + 1)$ and that of $(X,Z)$, $\rho^2_{XZ} = \sigma^2_X\Tilde{h}^2/(\sigma^2_X\Tilde{h}^2 + 1)$, the regions \eqref{gs-cp-gauss} and \eqref{cs-cp-gauss} can be transformed as
\begin{align}
\RG
    &=\bigcup_{0 < \alpha \le 1}\Big\{(R_S,R_J,R_L)\in \mathbb{R}^3_+: \nonumber \\
    &R_S \le\frac{1}{2}\log\frac{\alpha\rho^2_{XZ} + 1 - \rho^2_{XZ}}{\alpha\rho^2_{XY} + 1 - \rho^2_{XY}}, \nonumber \\
    &R_J \ge \frac{1}{2}\log\frac{\alpha\rho^2_{XY} + 1 - \rho^2_{XY}}{\alpha}, \nonumber \\
    &R_L \ge \frac{1}{2}\log\frac{\alpha\rho^2_{XY} + 1 - \rho^2_{XY}}{\alpha}\Big\} \label{special-1}, \displaybreak[0]\\
\RC
    &=\bigcup_{0 < \alpha \le 1}\Big\{(R_S,R_J,R_L)\in \mathbb{R}^3_+: \nonumber \\
    &R_S \le\frac{1}{2}\log\frac{\alpha\rho^2_{XZ} + 1 - \rho^2_{XZ}}{\alpha\rho^2_{XY} + 1 - \rho^2_{XY}}, \nonumber \displaybreak[0]\\
    &R_J \ge \frac{1}{2}\log\frac{\alpha\rho^2_{XZ} + 1 - \rho^2_{XZ}}{\alpha}, \nonumber \displaybreak[0]\\
    &R_L \ge \frac{1}{2}\log\frac{\alpha\rho^2_{XY} + 1 - \rho^2_{XY}}{\alpha}\Big\} \label{special-2}.
\end{align}
The regions in \eqref{special-1} and \eqref{special-2} align with \cite[Cor. 1]{vamoua-tifs} when \cite[eq. (5)]{vamoua-tifs} is replaced with \eqref{privacy} to eliminate the quantity $I(X;Z)$. Moreover, when the storage rate is not considered, i.e., \eqref{storage} is not imposed, and Eve has no side information, i.e., $\rho_{XZ} = 0$, the regions in \eqref{special-1} and \eqref{special-2} simplify to \cite[Th. 4.1 and 4.2]{iw-book}.

\subsection{Numerical Examples} \label{num-section}

{We begin by presenting numerical calculations that illustrate the relationship between the secret-key and storage rates in the GS model, and then proceed to compare the secret-key and privacy-leakage rates of the GS and CS models under the same storage rate}.

For investigating the relation of the secret-key and storage rates, we consider three cases, with the parameters summarized as follows:
1.~$\Omega_Y=\Omega_Z=1$ with $\H^\intercal_{k^*} = 0.95$ and $\tH^\intercal_{l^*} = 0.8$, 2.~$\Omega_Y=3$ and $\Omega_Z=1$ with $\H^\intercal_{k^*} = [0.95~0.95~0.95]$ and $\tH^\intercal_{l^*} = 0.8$, and 3.~$\Omega_Y=3$ and $\Omega_Z=4$ with $\H^\intercal_{k^*} = [0.95~0.95~0.95]$ and $\tH^\intercal_{l^*} = [0.8~0.8~0.5~0.5]$.
Moreover, we fix the variance of the source identifier as $\sigma^2_{X} = 5$ for all cases.

For a given $\alpha$, define the optimal storage rate $R_J(\alpha) = \frac{1}{2}\log\frac{\alpha\sigma^2_X\H^\intercal_{k^*}\H_{k^*} + 1}{\alpha(\sigma^2_X\H^\intercal_{k^*}\H_{k^*} + 1)},$ from which one can express $\alpha$ as
\begin{align}
\alpha = \frac{1}{2^{2R_J(\alpha)} +(2^{2R_J(\alpha)}-1)\sigma^2_X\H^\intercal_{k^*}\H_{k^*}}. \label{alpah-rj}
\end{align}
Substituting \eqref{alpah-rj} into the right-hand side of \eqref{rs-gauss-th1}, the optimal secret-key rate based on $R_J(\alpha)$ is given by
$
R_S(R_J(\alpha)) = \frac{1}{2}\log\left(\frac{\sigma^2_X\H^\intercal_{k^*}\H_{k^*}(1-2^{-2{R_J(\alpha)}}) + \sigma^2_X\tH^\intercal_{l^*}\tH_{l^*}2^{-2{R_J(\alpha)}} + 1}{\sigma^2_X\tH^\intercal_{l^*}\tH_{l^*} + 1 }\right)
$.
Note that if $R_J(\alpha) \rightarrow \infty$, $R^*_S(R_J(\alpha)) \rightarrow \frac{1}{2}\log\left(\frac{\sigma^2_X\H^\intercal_{k^*}\H_{k^*} + 1}{\sigma^2_X\tH^\intercal_{l^*}\tH_{l^*} + 1 }\right)$.

Figure \ref{fig-rjrs}(a) depicts the relation of $(R_J(\alpha),R_S({R_J(\alpha)}))$. In this figure, Case~2 (blue curve) shows a high secret-key rate compared to the other cases. This is due to an increase in the number of antennas at the decoder, which enhances the correlation between the source and observations at the terminal. On the other hand, in {Case~3} (red curve), as the number of antennas at Eve increases, the secret-key rate drops compared to Case~2 because the stronger correlation with Eve reduces the key-generation rate. Also, Case~3 shows that even when Eve has more antennas, a positive secret-key rate is still achievable as long as $\H^\intercal_{k^*}\H_{k^*} \ge \tH^\intercal_{l^*}\tH_{l^*}$.

{Figure \ref{fig-rjrs}(b) presents the secret-key and storage rates for a given $\alpha$, focusing on Case~3, where the maximum secret-key rate reaches 0.2771 (cf.~Fig.~\ref{fig-rjrs}(a)). As $\alpha \to 0$, the storage rate grows unbounded, reflecting the absence of encoding, while the secret-key rate is maximized. In contrast, as $\alpha \to 1$, both rates approach zero. According to \eqref{x-u-phi}, this is because $U$ is highly correlated with $X$ when $\alpha \to 0$, leading to a high secret-key rate, whereas $U$ becomes independent of $X$ when $\alpha = 1$, resulting in zero secret-key rate.}

{Figures \ref{fig-rjrs}(c) and \ref{fig-rjrs}(d) respectively compare the secret-key and privacy-leakage rates between the GS and CS models for Case~2, under the same values of storage rates. In the low storage rate regime, the GS model results in a higher secret-key rate than the CS model, but at the cost of greater privacy leakage, highlighting a trade-off between these two security metrics. In the high storage rate regime, both models achieve the same secret-key rate, but the GS model still incurs greater privacy leakage than the CS model with the difference equal to the secret-key rate. This occurs because, in the CS model, the concealed data (with rate equal to the secret-key rate) reveals no information about the source identifier. These results suggest that in practical system designs, where the storage space is fixed, the GS model may be preferred in the low storage rate regime when maximizing the secret-key rate is important, while the CS model is preferable for minimizing privacy leakage. In the high storage rate regime, the CS model becomes the preferred option as it can achieve the same secret-key rate as the GS model but with lower privacy leakage}.

\begin{figure*}[t]
  \centering
  \begin{minipage}[t]{0.245\linewidth}
  \centering
    \includegraphics[width=\linewidth]{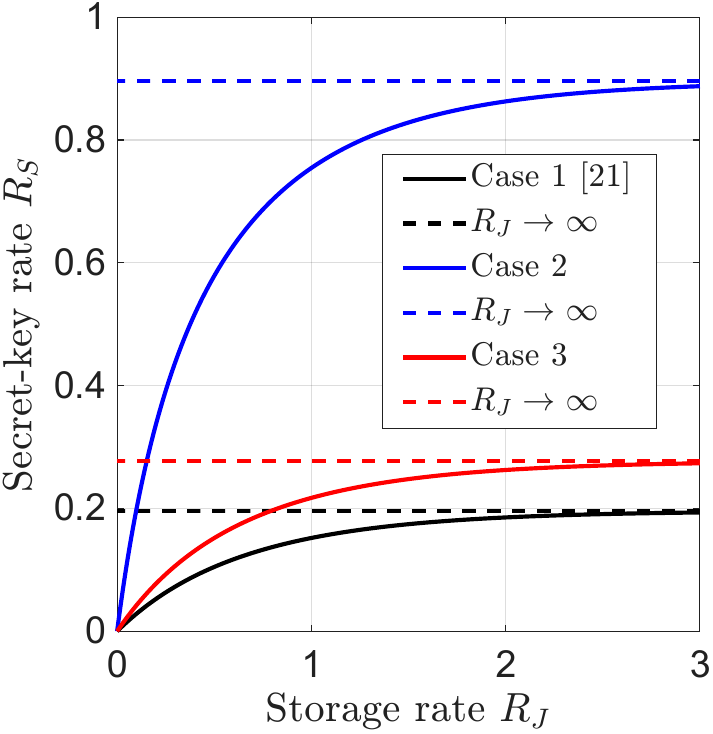}
    \small (a)
    \label{fig-a}
  \end{minipage}
  \hfill
  \begin{minipage}[t]{0.24\linewidth}
  \centering
    \includegraphics[width=\linewidth]{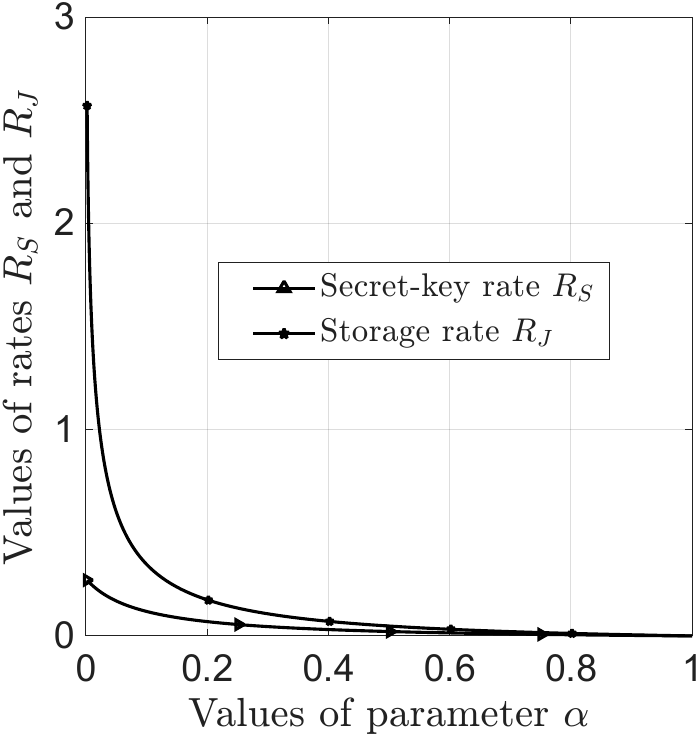}
    \small (b)
    \label{fig-b}
  \end{minipage}
  \hfill
  \begin{minipage}[t]{0.245\linewidth}
  \centering
    \includegraphics[width=\linewidth]{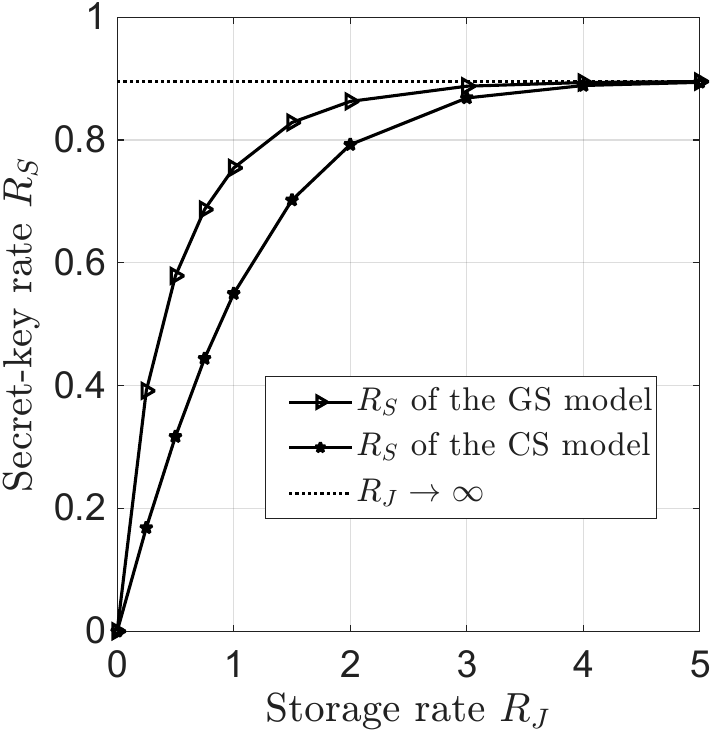}
    \small (c)
    \label{fig-c}
  \end{minipage}
  \hfill
  \begin{minipage}[t]{0.243\linewidth}
  \centering
    \includegraphics[width=\linewidth]{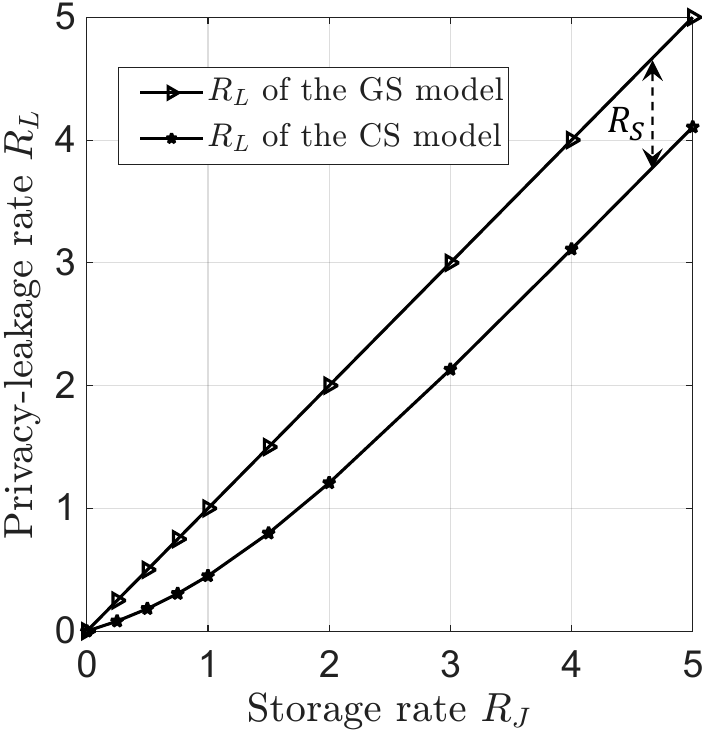}
    \small (d)
    \label{fig-d}
  \end{minipage}
  \caption{{(a) the relation of storage and secret-key rates in the GS model, (b) the secret-key and storage rates for a given value of $\alpha$ in the GS model, and for given a storage rate, a comparison of the secret-key and privacy-leakage rates in the GS and CS models are shown in (c) and (d), respectively}.}
  \label{fig-rjrs}
\end{figure*}

\section{Concluding Remarks and Future Directions} \label{sect4}

We studied secret-key generation from private identifiers under channel uncertainty and multiple-output settings. This setup addresses authentication robustness against {eavesdroppers in scenarios} where the legitimate terminals lack exact CSI and Eve may use multiple antennas to improve signal reception. {We} derived inner and outer bounds for discrete sources and provided a full capacity characterization for Gaussian sources. The main technical contributions lie in proving the inner bound for discrete memoryless sources and the outer bound for the Gaussian case.

To prove the inner bound for discrete sources, we first extend the technique used in \cite{liang-compound} for compound wiretap channels to ensure that the generated secret key is uniform and remains secret from Eve’s observation for any channel state. For the outer bound in the Gaussian case, we first employ sufficient statistics to convert the vector problem into a scalar one, so that we can use the degraded property of the scalar Gaussian random variables to derive a single-letter characterization of the outer region. Then, we apply the single-letter characterization to derive the parametric expression for the Gaussian case with Fisher information-based techniques playing a crucial role in the derivation. 

We also performed numerical evaluations for the Gaussian case to illustrate how changes in the number of antennas at the legitimate terminals and the eavesdropper affect the trade-offs between secret-key and storage rates, and to compare the secret-key and privacy-leakage rates of the GS and CS models under the same storage rate. The first set of results indicates that increasing the number of antennas at the decoder leads to a higher secret-key rate, while adding antennas at Eve reduces the secret-key rate. Nevertheless, even if Eve has more antennas, a positive secret-key rate remains achievable as long as the worst-case channel power gain at the decoder exceeds the best-case channel power gain at Eve. The second set of results shows that in the low storage-rate regime, the GS model achieves a higher secret-key rate, whereas the CS model offers better privacy-leakage performance. In contrast, in the high storage-rate regime, the CS model proves to be the more favorable choice as it provides the same secret-key rate as the GS model with lower privacy leakage.

A natural extension of this work is to characterize the capacity region for Gaussian sources under noisy enrollment channels. As noted in Lemma \ref{lemma-txneqx}, since we may not be able to model the covariance matrices of the independent noises as identity matrices, the analysis will become more involved compared to that of Theorem \ref{th1-gauss}. This arises because the scalar problem obtained by transforming the original vector problem using sufficient statistics results in more complicated forms than the expressions in Lemma \ref{bar-relation}.
Extending the scenario to the case of vector Gaussian sources is also an interesting topic. Another possible avenue is to include user identification as studied in \cite{itw-tit-2015,kitti-2016,vy2020} and see how the identification rate influences the capacity region.

\appendices
\section{Proof of Propositions \ref{inner-bound} and \ref{discrete-result-outer}}
\subsection{Proof of Proposition \ref{inner-bound}} \label{proof-inner}
We  only  prove \eqref{inner-gs} since \eqref{inner-cs} follows similarly with an extra procedure, a one-time pad procedure to conceal the chosen secret key. As a result, an extra rate equal to the secret-key rate is needed for storing the concealed key information in the database, which appears in the constraint of the storage rate of the CS model.

Fix the test channels $P_{U|\tX}$ and $P_{V|U}$ and let $\delta > 0$. In the following, we show that these rates are achievable
\begin{align}
    R_S&~\triangleq~\min_kI(\Y_k;U|V) - \max_lI(\Z_l;U|V) - \delta, \\
    R_J &~{\triangleq}~ \max_kI(\tX;U|V,\Y_k) + \max_kI(\tX;V|\Y_k) + 5\delta, \displaybreak[0]\\
    R_L &~{\triangleq}~ \max_kI(\tX;U|V,\Y_k) + \max_kI(\tX;V|\Y_k) \nonumber \\
    &~~~-I(\tX;U|X)+ \min_kI(\Y_k;V) -\min_lI(\Z_l;V) + 4\delta.
\end{align}
For the random codebook construction, we also define
\begin{align}
R_v&{\triangleq}~I(\tX;V) + \delta,~~{R_{J_{v_1}}}{\triangleq}\max_kI(\tX;V|\Y_k) + 2\delta, \\
 R_u&{\triangleq}~I(\tX;U|V) + \delta,~~{R_{J_{u_1}}}{\triangleq}\max_kI(\tX;U|\Y_k,V) + 3\delta,
\end{align}
and the sets $\mathcal{J}_{v_1}{\triangleq}[1:2^{n{R_{J_{v_1}}}}]$, $\mathcal{J}_{v_2}{\triangleq}[1:2^{n(R_v-{R_{J_{v_1}}})}]$, $\mathcal{J}_{u_1}{\triangleq}[1:2^{n{R_{J_{u_1}}}}]$, $\mathcal{J}_{u_2}{\triangleq}[1:2^{nR_S}]$, $\mathcal{J}_{u_3}{\triangleq}[1:2^{n(\max_lI(\Z_l;U|V)-\delta)}]$. Note that $R_u = {R_{J_{u_1}}} + R_S + \max_lI(\Z_l;U|V)-\delta$.

\smallskip
\noindent{\em \bf Random Codebook}: Generate i.i.d.\ sequences $v^n(j_{v_1},j_{v_2})$ from $P_{V^n}$, where $(j_{v_1},j_{v_2}) \in \mathcal{J}_{v_1}\times \mathcal{J}_{v_2}$.
For every $(j_{v_1},j_{v_2})$, generate i.i.d.\ sequences $u^n(j_{u_1},{j_{u_2},j_{u_3}},j_{v_1},j_{v_2})$, where $(j_{u_1},j_{u_2},j_{u_3}) \in \mathcal{J}_{u_1}\times  \mathcal{J}_{u_2}\times \mathcal{J}_{u_3}$, according to $P_{U^n|V^n=v^n(j_{v_1},j_{v_2})}$. All the generated sequences $(V^n(j_{v_1},j_{v_2}),U^n(j_{u_1},j_{u_2},j_{u_3},j_{v_1},j_{v_2}))$ form the~ codebook~$\mathcal{C}_n$. 

\smallskip
\noindent{\em \bf Encoding}: Observing $\tx^n{}$, the encoder first finds $(j_{v_1},j_{v_2})$ such that $(\tx^n,v^n(j_{v_1},j_{v_2})) \in \mathcal{T}^n_{\delta}$. Then, it looks for $(j_{u_1},j_{u_2},j_{u_3})$ such that $(\tx^n,u^n(j_{u_1},j_{u_2},j_{u_3},j_{v_1},j_{v_2})) \in \mathcal{T}^n_{\delta}(\tX U|v^n(j_{v_1},j_{v_2}))$. 
 If a unique tuple $(j_{u_1},j_{u_2},j_{u_3},j_{v_1},j_{v_2})$ is found, the encoder assigns the helper data  $j = (j_{v_1},j_{u_1})$ and the secret key $s=j_{u_2}$. If multiple such tuples are found, the encoder selects one tuple uniformly at random and assigns $j = (j_{v_1},j_{u_1})$ and $s=j_{u_2}$. In case no such tuple exists, the encoder sets all $j_{v_1}$, $j_{v_2}$, $j_{u_1}$, $j_{u_2}$, and $j_{u_3}$ to be one and assigns $j = (1,1)$ and $s=1$.

\smallskip
\noindent{\em \bf Decoding}: From $\y^n_k$ and $(j_{u_1},j_{v_1})$, the decoder first looks for the unique index $\hat{j}_{v_2}$ such that $(\y^n_k,v^n(j_{v_1},\hat{j}_{v_2})) \in \mathcal{T}^n_{\delta}$. Then, it looks for the unique pair $(\hat{j}_{u_2}, \hat{j}_{u_3})$ such that $(\y^n_k,u^n(j_{u_1},\hat{j}_{u_2},\hat{j}_{u_3},j_{v_1},\hat{j}_{v_2})) \in \mathcal{T}^n_{\delta}(\Y_kU|v^n(j_{v_1},\hat{j}_{v_2}))$. If the indices $\hat{j}_{u_2}, \hat{j}_{u_3}$, and $\hat{j}_{v_2}$ are uniquely determined, then the decoder estimates $\hat{s}=\hat{j}_{u_2}$; otherwise, it sets $\hat{s}=1$ and declares an error.

\smallskip
In the following, we write $V^n(J_{v_1},J_{v_2})$ and $U^n(J_{u_1},J_{u_2},J_{u_3},J_{v_1},J_{v_2})$ as $V^n$ and $U^n$ for convenience.

\smallskip
{\em Analysis of Error Probability}:~~~Possible error events at the encoder are
\begin{itemize}
    \item[$\mathcal{E}_1$]: $\{(\tX^n{},V^n(j_{v_1},j_{v_2}))~{\notin}~\mathcal{T}^n_{\delta}, \forall (j_{v_1},j_{v_2}) \in \mathcal{J}_{v_1}\times \mathcal{J}_{v_2}\}$,
    \item[$\mathcal{E}_2$]: $\{(\tX^n,U^n(j_{u_1},j_{u_2},j_{u_3},J_{v_1},J_{v_2}))\notin\mathcal{T}^n_{\delta}(\tX U|V^n)$, $ \forall(j_{u_1},j_{u_2},j_{u_3}) \in \mathcal{J}_{u_1}\times  \mathcal{J}_{u_2}\times \mathcal{J}_{u_3}\}$,
\end{itemize}
and those at the decoder are
\begin{itemize}
    \item[$\mathcal{E}_3$]: $\{(\Y^n_k,U^n,V^n) \notin \mathcal{T}^n_{\delta}\}$,
    \item[$\mathcal{E}_4$]: $\{ \exists j'_{v_2} \in \mathcal{J}_{v_2},j'_{v_2} \neq J_{v_2}$ and $(\Y^n_k,V^n(J_{v_1},j'_{v_2})) \in \mathcal{T}^n_{\delta}\}$,
    \item[$\mathcal{E}_5$]: $\{\exists (j'_{u_2},j'_{u_3}) \in \mathcal{J}_{u_2}\times \mathcal{J}_{u_3}, (j'_{u_2},j'_{u_3})\neq (J_{u_2},  J_{u_3})$ and $(\Y^n_k,U^n(J_{u_1},j'_{u_2},j'_{u_3},J_{v_1},J_{v_2}))\in \mathcal{T}^n_{\delta}(\Y_kU|V^n)\}$.
\end{itemize}
Then, we have
\begin{align}
    &\max_k{\mathbb{P}}\{\hat{S}_k \neq S\} = \mathbb{P}\{\cup_{i=1}^5\mathcal{E}_i\} \nonumber \\
    &\le \mathbb{P}\{\mathcal{E}_1\}+\mathbb{P}\{\mathcal{E}_2\} + \mathbb{P}\{\mathcal{E}_3\cap(\mathcal{E}_1\cup\mathcal{E}_2)^c\} + \mathbb{P}\{\mathcal{E}_4\} + \mathbb{P}\{\mathcal{E}_5\}.
\end{align}

The first and second terms vanish by the covering lemma \cite[Lemma 3.3]{GK} since $R_v > I(\tX;V)$ and $R_u > I(\tX;U|V)$, respectively. The third term vanishes by  Markov lemma \cite[Lemma 15.8.1]{cover}. The last two terms vanish by the packing lemma \cite[Lemma 3.1]{GK}, since the rate of index $\hat{j}_{v_2}$ is less than $\min_kI(\Y_k;V)$ and that of index pair $(\hat{j}_{u_2},\hat{j}_{u_3})$ is less than $\min_kI(\Y_k;U|V)$, respectively. Hence, we have
\vspace{-2mm}
\begin{align}
    \lim_{n \rightarrow \infty}\max_k{\mathbb{P}}\{\hat{S}_k \neq S\} \rightarrow 0. \label{error-achie}
\end{align}

Before we analyze the constraints \eqref{secretk}, \eqref{storage}, \eqref{secrecy}, and \eqref{privacy} in Definition \ref{def1}, we state two lemmas. The first one, Lemma \ref{enhanced-Eve}, is an extended version of \cite[Lemma A.1]{liang-compound} to incorporate conditional mutual information.
\begin{Lemma} \label{enhanced-Eve}
If the inequality $I(\Z_l;U|V) < I(\Z_{l'};U|V)$ holds for $l,l' \in \mathcal{L}$, then there exists a vector random variable $\A$ such that the equality $I(\Z_l,\A;U|V) = I(\Z_{l'};U|V)$ and the Markov chain $V-U-\tX-(\Z_l,\Z_{l'})-\A$ are satisfied.
\end{Lemma}
\begin{proof} Let $B$ be a binary random variable taking values $l$ and $l'$ with probabilities $p$ and $1-p$, respectively, where $ 0 \le p \le 1$, and assume that $B$ is independent of all other random variables. Define  $\A \triangleq (\Z_B,B)$ and  $\Gamma(p) = I(\Z_l,\A;U|V)$. Due to the independence of $B$, we have
\begin{align}
    \Gamma(p)= pI(\Z_l;U|V) + (1-p)I(\Z_l,\Z_{l'};U|V).
\end{align}
Now observe that $\Gamma(1) < I(\Z_{l'};U|V) \le \Gamma(0)$, where the first inequality follows by the assumption $I(\Z_l;U|V) < I(\Z_{l'};U|V)$. Due to the continuity of the function $\Gamma(p)$ for all $p \in [0,1]$, there exists a $p^* \in [0,1]$ such that $\Gamma(p^*) = I(\Z_{l'};U|V)$, and thus the equality $I(\Z_l,\A;U|V) = I(\Z_{l'};U|V)$ is satisfied with $\A = (\Z_{B^*},B^*)$ and $B^*$ taking the values $l$ and $l'$ with probabilities $p^*$ and $1-p^*$, respectively. Also, this choice of $\A$ ensures that the Markov chain $V-U-\tX-(\Z_l,\Z_{l'})-\A$ is satisfied.
\end{proof}

Lemma \ref{enhanced-Eve} is used to show the existence of a random variable that achieves $\max_{l}I(\Z_l;U|V)$ and forms a Markov chain with $(V,U,\tX)$, as detailed in the following intermediate step.

\smallskip
\emph{Intermediate Step}: For any $l \in \mathcal{L}$, by Lemma \ref{enhanced-Eve}, there exists $\A$ such that for
\begin{align}
    \tZ \triangleq (\Z_l,\A) ,\label{en-Eve}
\end{align}
we have $I(\tZ;U|V) = \max_{l}I(\Z_l;U|V)$ and
\begin{align}
V-U-\tX-\tZ. \label{ztile-mk}
\end{align}
Moreover, define a binary random variable $T$, which takes $1$ if $(U^n,\tX^n,\tZ^n) \in \mathcal{T}^n_{\delta}$ and $0$ otherwise. For large enough $n$, it holds that
\vspace{-2mm}
\begin{align}
P_T(1) \ge 1 - \Tilde{\delta}_n \label{pt-1}
\end{align}
with $\Tilde{\delta}_n \downarrow 0$ as $\delta \downarrow 0$ and $n \rightarrow \infty$. This follows because the pair $(U^n,\tX^n)$ is jointly typical with probability approaching one, as shown in \eqref{error-achie}, and $\tZ^n$ is i.i.d. generated according to $\prod_{t=1}^nP_{\tZ_t|\tX_t}$ from \eqref{ztile-mk}, and thus \eqref{pt-1} follows by applying the Markov lemma \cite[Lemma 15.8.1]{cover}. Similarly, we  have  joint typicality of $(V^n,\tZ^n)$ as $(V^n,\tX^n)$ is jointly typical with high probability. These properties are applied in proving the next lemma, which plays a key role in the analyses of the secret-key uniformity and secrecy-leakage.

\begin{Lemma} \label{huvzl}
For an arbitrary index $l \in \mathcal{L}$, we have
\begin{align}
H(J_{u_2}|J_{u_1},J_{v_1},\Z^n_l,\mathcal{C}_n) &\ge n(R_S - \xi_n), \label{hmzl}
\end{align}
where $\xi_n$ goes to zero as $\delta \downarrow 0$ and $n \rightarrow \infty$.
\end{Lemma}
\begin{proof} We have
\begin{align}
&H(J_{u_2}|J_{u_1},J_{v_1},\Z^n_l,\mathcal{C}_n) \nonumber \displaybreak[0]\\
    &\ge H(J_{u_2}|J_{u_1},J_{v_1},J_{v_2},\Z^n_l,\A^n,\mathcal{C}_n) \nonumber\displaybreak[0] \\
    &\overset{\rm (a)}= H(J_{u_2}|J_{u_1},J_{v_1},J_{v_2},\tZ^n,\mathcal{C}_n) \nonumber\displaybreak[0] \\
    &= H(J_{u_1},J_{u_2},J_{u_3},J_{v_1},J_{v_2},\tZ^n|\mathcal{C}_n)\nonumber \displaybreak[0]\\
    &~~~ - H(J_{u_3}|J_{u_1},J_{u_2},J_{v_1},J_{v_2},\tZ^n,\mathcal{C}_n) \nonumber\displaybreak[0] \\
    &~~~-H(J_{u_1},J_{v_1},J_{v_2},\tZ^n|\mathcal{C}_n) \nonumber \displaybreak[0]\\
    &\overset{\rm (b)}\ge H(J_{u_1},J_{u_2},J_{u_3},J_{v_1},J_{v_2},\tZ^n|\mathcal{C}_n) \nonumber \\
    &~~~-H(J_{u_1},J_{v_1},J_{v_2},\tZ^n|\mathcal{C}_n) -n\delta_n \nonumber \displaybreak[0] \\
    &\overset{\rm (c)}\ge H(U^n,\tZ^n|\mathcal{C}_n)-H(\tZ^n|V^n,\mathcal{C}_n) \nonumber \\
    &~~~- H(J_{u_1},J_{v_1},J_{v_2}|\mathcal{C}_n)  -n\delta_n \nonumber \displaybreak[0]\\
    &\ge~P_T(1)H(U^n,\tZ^n|T=1,\mathcal{C}_n)-H(\tZ^n|V^n,\mathcal{C}_n) \nonumber \\
    &~~~{-H(J_{u_1}|\mathcal{C}_n) -H(J_{v_1}|\mathcal{C}_n)-H(J_{v_2}|\mathcal{C}_n)-n\delta_n} \nonumber \displaybreak[0]\\
    &\overset{\rm (d)}\ge~(1- \Tilde{\delta}_n)H(U^n,\tZ^n|T=1,\mathcal{C}_n)-H(\tZ^n|V^n,\mathcal{C}_n) \nonumber \\
    &~~~{-H(J_{u_1}|\mathcal{C}_n) -H(J_{v_1}|\mathcal{C}_n)-H(J_{v_2}|\mathcal{C}_n)-n\delta_n} \nonumber \displaybreak[0]\\
    &\overset{\rm (e)}\ge n(1- \Tilde{\delta}_n)(I(\tX;U)+H(\tZ|U)-2\epsilon_{\delta})-H(\tZ^n|V^n,\mathcal{C}_n) \nonumber \\
    &~~~ - H(J_{u_1}|\mathcal{C}_n) - H(J_{v_1}|\mathcal{C}_n)-H(J_{v_2}|\mathcal{C}_n)  -n\delta_n \nonumber \displaybreak[0]\\
    &\overset{\rm (f)}\ge n(I(\tX;U)+H(\tZ|U)-H(\tZ|V)-\gamma_n) \nonumber \\
    &~~~ - H(J_{u_1}|\mathcal{C}_n) - H(J_{v_1}|\mathcal{C}_n)-H(J_{v_2}|\mathcal{C}_n)  - {n\delta'_n} \nonumber \displaybreak[0]\\
    &\overset{\rm (g)}\ge n(I(\tX;U)-I(\tZ;U|V)-\gamma_n) \nonumber \\ 
    &~~~- n(\max_kI(\tX;U|V,\Y_k) + 3\delta) \nonumber \\
    &~~~- n(\max_kI(\tX;V|\Y_k) + 2\delta) \nonumber \\
    &~~~-n(\min_kI(\Y_k;V) - \delta) - n\delta'_n \nonumber \displaybreak[0]\\
    &\overset{\rm (h)}= n(I(\tX;U)-I(\tZ;U|V)) \nonumber \\ 
    &~~~- n(\max_k\{I(\tX;U|V) - I(\Y_k;U|V)\}) \nonumber \\
    &~~~- n(\max_k\{I(\tX;V)-I(\Y_k;V)\})-n\min_kI(\Y_k;V)  \nonumber \\
    &~~~ -n(4\delta +\gamma_n + \delta'_n) \nonumber \displaybreak[0]\\
    &\overset{\rm (i)}=n(\min_kI(\Y_k;U|V) - \max_lI(\Z_l;U|V) -\delta - \xi_n) \nonumber \displaybreak[0]\\
    &=n(R_S - \xi_n),
\end{align}
where (a) holds from \eqref{en-Eve}, (b) follows because the index $J_{u_3}$ can be reliably estimated from $(J_{u_1},J_{u_2},J_{v_1},J_{v_2},\tZ^n)$, as  $\frac{1}{n}\log|\mathcal{J}_3| < \max_lI(\Z_l;U|V)=I(\tZ;U|V)$, (c) holds because $U^n$ and $V^n$ are determined by the tuple $(J_{u_1},J_{u_2},J_{u_3},J_{v_1},J_{v_2})$ and the pair $(J_{v_1},J_{v_2})$, respectively, (d) follows from \eqref{pt-1}, and (e) follows from
\begin{align}
    P_{\tZ^nU^n}(\tz^n,u^n) &\le \sum_{\tx^n \in \mathcal{T}^n_{\delta}(\tX|\tz^n,u^n)}P_{\tX^n\tZ^n}(\tx^n,\tz^n) \nonumber \displaybreak[0] \\
    &\le 2^{n(H(\tX|\tZ,U)-\epsilon_\delta)}\cdot 2^{-n(H(\tX,\tZ)-\epsilon_\delta)} \nonumber \displaybreak[0] \\
    &= 2^{-n(I(\tX;U) + H(\tZ|U)-2\epsilon_\delta)},
\end{align}
(f) follows because, as shown in the intermediate step, $(V^n,\tZ^n)$ is jointly typical with high probability and thus  $H(\tZ^n|V^n,\mathcal{C}_n) \le n(H(\tZ|V) + \gamma_n)$ (cf. \cite[eq. (16)]{remi-2014}) and $\delta'_n \triangleq \Tilde{\delta}_n(I(\tX;U)+H(\tZ|U))+2(1 -\Tilde{\delta}_n)\epsilon_{\delta} + \delta_n$, (g) is due to the Markov chain \eqref{ztile-mk} and $H(J_{u_1}|\mathcal{C}_n) \le n{R_{J_{u_1}}}$, $H(J_{v_1}|\mathcal{C}_n) \le~n{R_{J_{v_1}}}$, $H(J_{v_2}|\mathcal{C}_n) \le n(R_v-R_{J_{v_1}}) = n(\min_kI(\Y_k;V) - \delta)$, (h) is due to the Markov chain $V-U-\tX-\Y_k$,  (i) follows from $I(\tZ;U|V) = \max_lI(\Z_l;U|V)$ and $\xi_n~\triangleq~3\delta +\gamma_n + {\delta'_n}$.
\end{proof}

{\em Analyses of Uniformity and Secrecy-leakage}:~~~The constraints of \eqref{secretk} and \eqref{secrecy} can be evaluated as
\begin{align}
    H(S|\mathcal{C}_n) &= H(J_{u_2}|\mathcal{C}_n) \nonumber \\
    &\ge H(J_{u_2}|J_{u_1},J_{v_1},\Z^n_l,\mathcal{C}_n) \nonumber \displaybreak[0]\\
    &\ge n(R_S - \xi_n), \label{rs-finall}
\end{align}
and
\begin{align}
\max_l&I(S;J,\Z^n_l|\mathcal{C}_n)=\max_lI(J_{u_2};J_{u_1},J_{v_1},\Z^n_l|\mathcal{C}_n) \nonumber  \displaybreak[0]\\
    &= \max_l\{H(J_{u_2}|\mathcal{C}_n)-H(J_{u_2}|J_{u_1},J_{v_1},\Z^n_l,\mathcal{C}_n)\} \nonumber \displaybreak[0]\\
    &\le \max_l\{nR_S -n(R_S - \xi_n)\} = n\xi_n, \label{secrecy-rinal}
\end{align}
where \eqref{rs-finall} and \eqref{secrecy-rinal} follow from Lemma \ref{huvzl}.

\smallskip
{\em Analysis of Storage Rate}:~~~The helper data is $J = (J_{v_1},J_{u_1})$, and thus the total storage rate is 
$
 \frac{1}{n}\log|\mathcal{J}_{v_1}||\mathcal{J}_{u_1}| = R_{J_{v_1}} + R_{J_{u_1}} = R_J.
$

\smallskip
{\em Analysis of Privacy-Leakage Rate}:
We have
\begin{align}
&\max_lI(X^n;J|\Z^n_l,\mathcal{C}_n)=\max_lI(X^n;J_{u_1},J_{v_1}|\Z^n_l,\mathcal{C}_n)\nonumber \\ 
&= I(X^n;J_{u_1},J_{v_1}|\mathcal{C}_n) -\min_lI(\Z^n_l;J_{u_1},J_{v_1}|\mathcal{C}_n) \nonumber \displaybreak[0]\\
    &\le H(J_{u_1}|\mathcal{C}_n) + H(J_{v_1}|\mathcal{C}_n) - H(J_{u_1},J_{v_1}|X^n,\mathcal{C}_n) \nonumber \\
    &~~~-\min_lI(\Z^n_l;J_{u_1},J_{v_1}|\mathcal{C}_n) \nonumber \displaybreak[0]\\
    &= H(J_{u_1}|\mathcal{C}_n) + H(J_{v_1}|\mathcal{C}_n) - H(\tX^n,J_{u_1},J_{v_1}|X^n,\mathcal{C}_n) \nonumber \\
    &~~~+ H(\tX^n|X^n,J_{u_1},J_{v_1},\mathcal{C}_n) -\min_lI(\Z^n_l;J_{u_1},J_{v_1}|\mathcal{C}_n) \nonumber \displaybreak[0]\\
   &\overset{\rm (a)}\le H(J_{u_1}|\mathcal{C}_n) + H(J_{v_1}|\mathcal{C}_n) - nH(\tX|X)\nonumber \\
    &~~~+n(H(\tX|X,U) + \epsilon'_{n})-\min_lI(\Z^n_l;J_{u_1},J_{v_1}|\mathcal{C}_n) \nonumber \displaybreak[0]\\
    &\le H(J_{u_1}|\mathcal{C}_n) + H(J_{v_1}|\mathcal{C}_n) - nI(\tX;U|X) \nonumber \\
    &~~~-\min_lI(\Z^n_l;J_{v_1}|\mathcal{C}_n) + n\epsilon'_n \nonumber \displaybreak[0]\\
     &\le H(J_{u_1}|\mathcal{C}_n) + H(J_{v_1}|\mathcal{C}_n) - nI(\tX;U|X) \nonumber \displaybreak[0]\\
    &~~~-{\min_l\{H(\Z^n_l)-H(\Z^n_l|J_{v_1},J_{v_2},\mathcal{C}_n)\}} \nonumber \\ 
    &~~~+\max_lI(J_{v_2};\Z^n_l|J_{v_1},\mathcal{C}_n) + n\epsilon'_n \nonumber \displaybreak[0]\\
   &\overset{\rm (b)}\le H(J_{u_1}|\mathcal{C}_n) + H(J_{v_1}|\mathcal{C}_n) - nI(\tX;U|X) \nonumber \\
    &~~~-{\min_l\{H(\Z^n_l)-H(\Z^n_l|V^n,\mathcal{C}_n)\}}+H(J_{v_2}|\mathcal{C}_n) + n\epsilon'_n \nonumber \displaybreak[0]\\
    &\overset{\rm (c)}\le H(J_{u_1}|\mathcal{C}_n) + H(J_{v_1}|\mathcal{C}_n) -nI(\tX;U|X)\nonumber \\
    &~~~-n(\min_l\{H(\Z_l) - H(\Z_l|V)\}) + H(J_{v_2}|\mathcal{C}_n)+ n\epsilon''_n \nonumber \displaybreak[0]\\
    &\le n(\max_k I(\tX;V|\Y_k) + \max_kI(\tX;U|V,\Y_k) + 4\delta \nonumber \\
    &~~~- I(\tX;U|X)+ \min_kI(\Y_k;V) - \min_lI(\Z_l;V) + \epsilon''_n) \nonumber \displaybreak[0]\\
    &= n(R_L + \epsilon''_n), \label{rl-final-achie}
\end{align}
where (a) follows from \eqref{goal-test}, shown below, and the codebook $\mathcal{C}_n$ is independent of $(\tX^n,X^n,\Y^n_k,\Z^n_l)$, (b) holds because $V^n$ is a function of $(J_{v_1},J_{v_2})$, and (c) follows because $H(\Z^n_l|V^n,\mathcal{C}_n) \le n(H(\Z_l|V) + \gamma'_n)$ and $\epsilon''_n ~\triangleq~ \gamma'_n + \epsilon'_n$. For brevity, define $E~\triangleq~(J_{u_2},J_{u_3},J_{v_2})$, where the decoder can reliably estimate the index $E$ for given $(\Y^n_k,J_{u_1},J_{v_1})$. Observe that
\begin{align}
&H(\tX^n|X^n,J_{u_1},J_{v_1},\mathcal{C}_n) \nonumber \displaybreak[0]\\ 
    &= H(\tX^n|X^n,J_{u_1},J_{v_1},E,\mathcal{C}_n) + I(E;\tX^n|X^n,J_{u_1},J_{v_1},{\mathcal{C}_n}) \nonumber \displaybreak[0]\\
    &\overset{\rm (a)}\le H(\tX^n|X^n,U^n,\mathcal{C}_n)+ H(E|X^n,J_{u_1},J_{v_1},{\mathcal{C}_n}) \nonumber \displaybreak[0]\\
    &\overset{\rm (b)}\le H(\tX^n|X^n,U^n,\mathcal{C}_n) + H(E|\Y^n_k,J_{u_1},J_{v_1},{\mathcal{C}_n}) \nonumber \displaybreak[0]\\
    &\overset{\rm (c)}\le H(\tX^n|X^n,U^n,\mathcal{C}_n) + n\epsilon_n \nonumber \displaybreak[0]\\
    &\overset{\rm (d)}\le n(H(\tX|X,U) + \epsilon'_n), \label{goal-test}
\end{align}
where (a) follows since $U^n$ is a function of $(J_{u_1},J_{u_2},J_{u_3},J_{v_1},J_{v_2})$, (b) is due to the Markov chain $E-(X^n,J_{u_1},J_{v_1})-\Y^n_k$ and conditioning reduces entropy, (c) follows from Fano's inequality with $\epsilon_n\downarrow 0$ as $\delta\downarrow 0$ and $n\rightarrow \infty$, and (d) because $H(\tX^n|X^n,U^n,\mathcal{C}_n) \le n(H(\tX|X,U) + \gamma''_n)$ for jointly typical sequences (cf. \cite[eq. (16)]{remi-2014}) and $\epsilon'_n ~{\triangleq}~ \gamma''_n+\epsilon_n$.

From \eqref{error-achie}, \eqref{rs-finall}, \eqref{secrecy-rinal}, and \eqref{rl-final-achie}, there must be at least one codebook satisfying all the conditions in Definition \ref{def1}, so that the region in \eqref{inner-gs} is achievable. 

\subsection{Proof of Proposition \ref{discrete-result-outer}} \label{proof-outer}

The cardinality bounds of the auxiliary random variables can be obtained from the support lemma \cite[Lemma 3.4]{GK}.

We show the proof for the GS model via a result derived in \cite{gksc2018} in the absence of action cost. By replacing \eqref{privacy} with $I(X^n;J,\Z^n_l) \le n(R_L + \delta)$, the outer bound for a pair $(k,l)$, denoted by $\mathcal{O}_{G_{kl}}$, is \cite[Th. 3]{gksc2018}
\begin{align}
    \mathcal{O}_{G_{kl}} &~{\triangleq}\bigcup_{P_{V|U},P_{U|\tX}}\{(R_S,R_J,R_L)\in \mathbb{R}^3_+:~\nonumber \\
    R_S &\le I(\Y_k;U|V) - I(\Z_l;U|V),~
    R_J \ge I(\tX;U|{\Y_k}), \nonumber \\
    R_L &\ge I(X;U,\Y_k) - I(X;\Y_k|V)+I(X;\Z_l|V)\},
    \label{outer-gunlu-gs}
\end{align}
where $V$ and $U$ satisfy $V-U-\tX-X-(\Y_k,\Z_l)$.

In Definition \ref{def1}, the constraints on the secret-key and storage rates are the same as in \cite[Def. 6]{gksc2018}, and the resulting bounds are given in the same form as in \eqref{outer-gunlu-gs}. For the privacy-leakage rate, we expand the right-hand side of $R_L$ in \eqref{outer-gunlu-gs}~as
\begin{align}
    &I(X;U,\Y_k)-I(X;\Y_k|V)+I(X;\Z_l|V) \nonumber \\
    & = I(X;U|\Y_k) + I(\Y_k;V)- I(\Z_l;V) + I(X;\Z_l),
    \label{rl-lower}
\end{align}
where we use the Markov chain $V-X-(\Y_k,\Z_l)$. By \eqref{rl-lower}, the privacy-leakage rate is lower bounded as
\begin{align}
    &n(R_L + \delta) \ge I(X^n;J|\Z^n_l) = I(X^n;J,\Z^n_l) - I(X^n;\Z^n_l) \nonumber \displaybreak[0] \\
    &\ge n(I(X;U|\Y_k) + I(\Y_k;V) -I(\Z_l;V)). \label{goal-222}
\end{align}
Therefore, an outer bound for a given $(k,l)$ in Definition \ref{def1} is
\begin{align}
    \mathcal{O}_{G_{kl}}
    ~{\triangleq}&\bigcup_{P_{V|U},P_{U|\tX}}\{(R_S,R_J,R_L)\in \mathbb{R}^3_+:~\nonumber \\
    R_S &\le I(\Y_k;U|V) - I(\Z_l;U|V),~
    R_J \ge I(\tX;U|\Y_k), \nonumber \\
    R_L &\ge I(X;U|\Y_k) + I(\Y_k;V) -I(\Z_l;V)\},
    \label{outer-k}
\end{align}
where $V$ and $U$ satisfy $V-U-\tX-X-(\Y_k,\Z_l)$.  Hence,
\begin{align}
    \RG \subseteq &\bigcap_{k \in \mathcal{K}}\bigcap_{l \in \mathcal{L}} \mathcal{O}_{G_{kl}} .\label{outer-for-all}
\end{align}

The proof for the CS model can be derived using the same reasoning from \cite[Th. 4]{gksc2018}.

\section{Proof of Lemma \ref{lemma-111}} \label{proof-lemma111}
Denote the covariance matrix of $(X,\Y_k,\Z_l)$ as $\bSigma$, where
\begin{align}
    \bSigma = \begin{bmatrix}
        \sigma^2_X&\bSigma_{X\Y_k} & \bSigma_{X\Z_l}\\
        \bSigma_{\Y_kX}& \bSigma_{\Y_k}& \bSigma_{\Y_k\Z_l} \\
        \bSigma_{\Z_lX} & \bSigma_{\Z_l\Y_k} & \bSigma_{\Z_l}
    \end{bmatrix}. \label{xyz-cova}
\end{align}

Note that \eqref{errorp} depends on the marginal distribution of $(X,\Y_k)$, \eqref{secretk} depends on the marginal distribution of $X$, and \eqref{secrecy} and \eqref{privacy} depend on the marginal distribution of $(X,\Z_l)$. Therefore, without loss of generality, using \cite[Th. 3.5.2]{gallager-2013} and \eqref{xyz-cova}, it suffices to consider
\begin{align}
    \Y_k &= \bSigma_{\Y_kX}\sigma^{-2}_XX + \N_{\Y_k}, \label{yk-cv} \\
    \Z_l &= \bSigma_{\Z_lX}\sigma^{-2}_XX + \N_{\Z_l}, \label{zl-cv}
\end{align}
where $\N_{\Y_k} \sim \mathcal{N}(0,\bSigma_{\N_{\Y_k}})$ with $\bSigma_{\N_{\Y_k}} = \bSigma_{\Y_k}-\bSigma_{\Y_k X}\sigma^{-2}_X\bSigma_{X\Y_k}$, and $\N_{\Z_l} \sim \mathcal{N}(0,\bSigma_{\N_{\Z_l}})$ with $\bSigma_{\N_{\Z_l}} = \bSigma_{\tX}-\bSigma_{\Z_l X}\sigma^{-2}_X\bSigma_{X\Z_l}$, independent of $X$.

Since $\bSigma$ is non-singular (positive definite), the sub-matrix $\begin{bmatrix}
\sigma^2_X&\bSigma_{X\Y_k} \\
\bSigma_{\Y_kX}& \bSigma_{\Y_k}
\end{bmatrix}$ is also positive definite. This implies that the matrix $\bSigma_{\N_{\Y_k}}$ is positive definite as it is the Schur complement of
$\sigma^2_X$ in the sub-matrix.
By Cholesky decomposition, there exists an invertible matrix ${\bf C} \in \mathbb{R}^{\Omega_Y\times\Omega_Y}$ such that $\bSigma_{\N_{\Y_k}} = {\bf C}{\bf C}^\intercal$. Then, we can reformulate \eqref{yk-cv} as
\begin{align}
    \Y'_k = {\bf A}_{\Y'_k}X + \N'_{\Y_k}, \label{tx-cv}
\end{align}
where $\Y'_k = {\bf C}^{-1}\Y_k$, ${\bf A}_{\Y'_k} = {\bf C}^{-1}\bSigma_{\Y_kX}\sigma^{-2}_X$ and $\N'_{\Y_k} \sim \mathcal{N}({\bf 0},{\bf I}_{\Omega_Y})$. Similarly, the same approach can be applied to~\eqref{zl-cv}.

\section{Proof of Theorem \ref{th1-gauss}} \label{th1-proof}
This appendix consists of two parts, that is, the achievability part in Appendix \ref{appendix-3A} and the converse part in Appendix \ref{appendix-3B}.

\subsection{Achievability Proof} \label{appendix-3A}
Note that Proposition \ref{inner-bound} was proved under finite source alphabets. However, the result can be extended to Gaussian sources as well by employing a fine quantization before encoding and decoding processes, similar to \cite{rana2022}.

By choosing $V$ as a constant in Proposition \ref{inner-bound}, the following regions are achievable.
\begin{align}
R_G~\supseteq~\bigcup_{P_{U|X}}&\Big\{(R_S,R_J,R_L)\in \mathbb{R}^3_+: \nonumber \\
    R_S &\le \min_{k}I(\Y_k;U)-\max_{l}I(\Z_l;U),\nonumber \\
    R_J &\ge I(X;U)-\min_kI(\Y_k;U), \nonumber \\
    R_L &\ge I(X;U)-\min_{l}I(\Y_k;U)\Big\},
    \label{lemma1-gs-gauss} \displaybreak[0] \\
R_C~\supseteq~\bigcup_{P_{U|X}}&\Big\{(R_S,R_J,R_L)\in \mathbb{R}^3_+: \nonumber \\
    R_S &\le \min_{k}I(\Y_k;U)-\max_{l}I(\Z_l;U), \nonumber \\
    R_J &\ge I(X;U)-\max_lI(\Z_l;U),\nonumber \\
    R_L &\ge I(X;U)-\min_{l}I(\Y_k;U)\Big\},
    \label{lemma1-cs-gauss}
\end{align}
where auxiliary random variable $U$ satisfies the Markov chain $U-X-(\Y_k,\Z_l)$ for all $k \in \mathcal{K}$ and $l \in \mathcal{L}$.

\begin{Lemma}[Weinstein–Aronszajn Identity {\cite[Appx. B]{Pozrikidis2014}}] \label{wai-lemma} For any $a \in \mathbb{R}_+$ and matrix $\H \in \mathbb{R}^{\Omega\times1}$, we have
\begin{align}
    {\det}(\H a\H^\intercal + {\bf I}_{\Omega}) = a\H^\intercal\H + 1,
\end{align}
where ${\det}(\cdot)$ denotes the determinant of a matrix.
\end{Lemma}

\begin{Lemma} \label{lemma-fx}
For given $\alpha \in (0,1]$, the function
\begin{align}
    f(\H^\intercal\H) = \log\left(\frac{\sigma^2_X\H^\intercal\H + 1}{\alpha\sigma^2_X\H^\intercal\H + 1}\right)
\end{align}
is monotonically increasing with respect to $\H^\intercal\H$.
\end{Lemma}

Lemma \ref{wai-lemma} is applied in the calculation of mutual information with vector random variables for given $k$ and $l$, and the role of Lemma \ref{lemma-fx} is to find the minimum and maximum values of the mutual information among all possible $k \in \mathcal{K}$ and $l \in \mathcal{L}$.

For $0< \alpha \le 1$, consider
\begin{align}
   X ~{\triangleq}~ U + \Theta, \label{x-u-phi}
\end{align}
where $U \sim \mathcal{N}(0,(1-\alpha)\sigma^2_X))$ and $\Theta \sim \mathcal{N}(0,\alpha\sigma^2_X)$. This relation implies that
\begin{align}
    I(X;U) = \frac{1}{2}\log\left(\frac{1}{\alpha}\right). \label{itxuuu}
\end{align}
From \eqref{channel-eq1} and \eqref{x-u-phi}, it follows that
\begin{align}
    \Y_k &= \H_k U + \H_k\Theta + \N_{\Y_k}, \\
    \Z_l &= \tH_l U + \tH_l\Theta + \N_{\Z_l}.
\end{align}
Using Lemma \ref{wai-lemma}, we have
\begin{align}
    I(\Y_k;U) &=\frac{1}{2}\log \frac{\sigma^2_X\H_k^\intercal\H_k + 1}{\alpha\sigma^2_X\H_k^\intercal\H_k + 1}, \displaybreak[0] \\
    I(\Z_l;U) &= \frac{1}{2}\log \frac{\sigma^2_X\tH^\intercal_l\tH_l + 1}{\alpha\sigma^2_X\tH^\intercal_l\tH_l + 1}
\end{align}
for a fixed pair $(k,l)$, and invoking Lemma \ref{lemma-fx} gives
\begin{align}
    \min_k I(\Y_k;U) &= \frac{1}{2}\log \left(\frac{\sigma^2_X\H^\intercal_{k^*}\H_{k^*} + 1}{\alpha\sigma^2_X\H^\intercal_{k^*}\H_{k^*} + 1}\right), \label{min-iyu-achie} \displaybreak[0]\\
    \max_l I(\Z_l;U) &= \frac{1}{2}\log \left(\frac{\sigma^2_X\tH^\intercal_{l^*}\tH_{l^*} + 1}{\alpha\sigma^2_X\tH^\intercal_{l^*}\tH_{l^*} + 1}\right). \label{max-izu-achie}
\end{align}
Finally, substituting \eqref{itxuuu}, \eqref{min-iyu-achie}, \eqref{max-izu-achie} into \eqref{lemma1-gs-gauss} and \eqref{lemma1-cs-gauss}, gives \eqref{gs-cp-gauss} and \eqref{cs-cp-gauss}.

\subsection{Converse Proof} \label{appendix-3B}

We will need the following lemmas. These lemmas convert vector observations in \eqref{channel-eq1} into scalar Gaussian random variables using sufficient statistics \cite[Sect. 2.9]{cover}. This transformation plays an important role in deriving the outer bound of Gaussian sources.
\begin{Lemma}[\hspace{-0.1mm}{\cite[Lemma 3.1]{parada-isit2005}}] \label{scalar-lemma-CB} Consider a channel with input $W$ and output $\Tilde{\bf W}$, namely,
$
\Tilde{\bf W}~{\triangleq}~{\bf A}W + {\bf N}_{\Tilde{\bf W}},
$
where $\A$ is a matrix and ${\bf N}_{\Tilde{\bf W}} \sim \mathcal{N}({\bf 0},\bSigma_{\Tilde{\bf W}})$. A sufficient statistic to correctly determine $W$ from $\Tilde{\bf W}$ is the following scalar
\begin{align}
    \bar{W}~{\triangleq}~{\bf A}^\intercal{\bf \Sigma}^{-1}_{\Tilde{\bf W}}\Tilde{\bf W}.
\end{align}
\end{Lemma}

\begin{Lemma} \label{bar-relation} The vector equations in \eqref{channel-eq1} can be rewritten as 
\begin{align}
    \bar{Y}_k = \nu_{\bar{Y}_k}X + N_{\bar{Y}_k}, \bar{Z}_l = \nu_{\bar{Z}_l}X + N_{\bar{Z}_l}, \label{scalar-gaussian}
\end{align}
where $\nu_{\bar{Y}_k} \triangleq \H^\intercal_k\H_k$, $\nu_{\bar{Z}_l} \triangleq \tH^\intercal_l\tH_l$, $N_{\bar{Y}_k} \sim \mathcal{N}(0,\nu_{\bar{Y}_k})$, and $N_{\bar{Z}_l} \sim \mathcal{N}(0,\nu_{\bar{Z}_l})$.
\end{Lemma}
\begin{proof}
Applying Lemma \ref{scalar-lemma-CB} to our settings in \eqref{channel-eq1}, we have
\begin{align}
\bar{Y}_k~=~\H^\intercal_k{\bf I}^{-1}_{\Omega_Y}\Y_k,~ \bar{Z}_l~=~\tH^\intercal_l{\bf I}^{-1}_{\Omega_Z}\Z_l. \label{barxyz}
\end{align}
Now substituting \eqref{channel-eq1} into \eqref{barxyz}, we have
\begin{align}
    \bar{Y}_k &= \H^\intercal_k(\H_k X + \mathbf{N}_{\Y_k}) = \nu_{\bar{Y}_k}X + N_{\bar{Y}_k}, \label{bary-gauss} \\
    \bar{Z}_l & = \tH^\intercal_l(\tH_l X + \mathbf{N}_{\Z_l}) = \nu_{\bar{Z}_l}X + N_{\bar{Z}_l}, \label{barz-gauss}
\end{align}
where we denote $\nu_{\bar{Y}_k}~{\triangleq}~\H^\intercal_k\H_k$, $\nu_{\bar{Z}_l} \triangleq \tH^\intercal_l\tH_l$, $N_{\bar{Y}_k}~{\triangleq}~\H_k\mathbf{N}_{\Y_k}$, and $N_{\bar{Z}_l}~{\triangleq}~\tH_l\mathbf{N}_{\Z_l}$. 
Note that $N_{\bar{Y}_k}$ and $N_{\bar{Z}_l}$ are Gaussian random variables, and their 
variances are ${\rm Var}[\H^\intercal_k\mathbf{N}_{\Y_k}] = \H^\intercal_k\H_k$, and ${\rm Var}[\tH^\intercal_l\mathbf{N}_{\Z_l}] = \tH^\intercal_l\tH_l$.
\end{proof}

As $(X,\bar{Y}_k,\bar{Z}_l)$ are scalar Gaussian random variables, when the squared value of the correlation coefficient of $(X,\bar{Y}_{k^*})$ is greater than that of $(X,\bar{Z}_{l^*})$, i.e., $\H^\intercal_{k^*}\H_{k^*} \ge \tH^\intercal_{l^*}\tH_{l^*}$, implying that $\H^\intercal_{k}\H_{k} \ge \tH^\intercal_{l}\tH_{l}$ for any pair $(k,l)$, there exist Gaussian random variables $(X',\bar{Y}'_k,\bar{Z}'_l)$ such that $X'-\bar{Y}'_k-\bar{Z}'_l$ is satisfied, and the marginal distributions of $(X,\bar{Y}_k)$ and $(X',\bar{Y}'_k)$ and that of $(X,\bar{Z}_l)$ and $(X',\bar{Z}'_l)$ coincide \cite[Lemma 6]{wataoha2010}. In the subsequent discussions, we denote $(X',\bar{Y}'_k,\bar{Z}'_l)$ as $(X,\bar{Y}_k,\bar{Z}_l)$ for brevity. This  property is used in the derivation of the following~lemma.
\begin{Lemma}[Outer bounds] \label{lemma7-outer}
If $\H^\intercal_{k^*}\H_{k^*} \ge \tH^\intercal_{l^*}\tH_{l^*}$, the outer bounds of the GS and CS models are provided as
\begin{align}
    \RG~&\subseteq~\bigcap_{k \in \mathcal{K}}\bigcap_{l \in \mathcal{L}}\bar{\mathcal{O}}_{G_{kl}},~~~
    \RC\subseteq~\bigcap_{k \in \mathcal{K}}\bigcap_{l \in \mathcal{L}}\bar{\mathcal{O}}_{C_{kl}},
    \label{discrete-outer-cs-gauss}
\end{align}
where $\bar{\mathcal{O}}_{G_{kl}}$ and $\bar{\mathcal{O}}_{C_{kl}}$ are outer bounds of the GS and CS models for a given pair $(k,l)$ and are defined as
\begin{align}
\bar{\mathcal{O}}_{G_{kl}} \triangleq \bigcup_{P_{U|X}}&\Big\{(R_S,R_J,R_L)\in \mathbb{R}^3_+:~\nonumber \\
    R_S &\le I(\bar{Y}_k;U)-I(\bar{Z}_l;U), \nonumber \\
    R_J &\ge I(X;U)-I(\bar{Y}_k;U), \nonumber \\
    R_L &\ge I(X;U)-I(\bar{Y}_k;U)\Big\}, \label{o-g-kl} \displaybreak[0]\\
\bar{\mathcal{O}}_{C_{kl}} \triangleq \bigcup_{P_{U|X}}&\Big\{(R_S,R_J,R_L)\in \mathbb{R}^3_+:~\nonumber \\
    R_S &\le I(\bar{Y}_k;U)-I(\bar{Z}_l;U), \nonumber \\
    R_J &\ge I(X;U)-I(\bar{Z}_l;U), \nonumber \\
    R_L &\ge I(X;U)-I(\bar{Y}_k;U)\Big\}, \label{o-c-kl}
\end{align}
and $U$ satisfies the Markov chain $U-X-\bar{Y}_k-\bar{Z}_l$. If $\H^\intercal_{k^*}\H_{k^*} < \tH^\intercal_{l^*}\tH_{l^*}$, $
\RG = \RC = \{(R_S,R_J,R_L):
    R_S = 0, R_J \ge 0, R_L \ge 0\}.
$
\end{Lemma}
\begin{proof}
The proof is provided in Appendix~\ref{proof-lemma7} and follows standard converse proof techniques, where Fano's inequality and the introduction of auxiliary random variable are used. The key idea is to exploit the relationship in~\eqref{barxyz}, which shows that $(\bar{Y}_k, \bar{Z}_l)$ and $(\Y_k, \Z_l)$ are mutually deterministic. This allows the vector random variables $(\Y_k, \Z_l)$ to be removed from the rate constraints during the analysis.
\end{proof}

Next, we utilize the single-letter expressions in Lemma \ref{lemma7-outer} to derive the parametric forms for Gaussian sources. We begin with the proof of \eqref{o-g-kl}. Each rate constraint in \eqref{o-g-kl} can be expanded as
\vspace{-2mm}
\begin{align}
R_S &\le I(\bar{Y}_k;U)-I(\bar{Z}_l;U)
    \overset{\rm (a)}= \frac{1}{2}\log\frac{\H^\intercal_k\H_k(\sigma^2_X\H^\intercal_k\H_k + 1)}{\tH^\intercal_l\tH_l(\sigma^2_X\H^\intercal_l\H_l + 1)} \nonumber \\
    &~~~~+ h(\bar{Z}_l|U)- h(\bar{Y}_k|U), \label{label-rs} \\
R_J &\ge I(X;U)-I(\bar{Y}_k;U)
    \overset{\rm (b)}= \frac{1}{2}\log\frac{\sigma^2_X}{\H^\intercal_k\H_k(\sigma^2_X\H^\intercal_k\H_k + 1)} \nonumber \\ &~~~~+ h(\bar{Y}_k|U) - h(X|U), \label{label-rj} \\
R_L &\ge I(X;U)-I(\bar{Y}_k;U)
    \overset{\rm (c)}= \frac{1}{2}\log\frac{\sigma^2_X}{\H^\intercal_k\H_k(\sigma^2_X\H^\intercal_k\H_k + 1)} \nonumber \\
    &~~~~+ h(\bar{Y}_k|U) - h(X|U), \label{label-rl}
\end{align}
where (a), (b), and (c) follow from \eqref{scalar-gaussian}.

From Lemma \ref{bar-relation}, we have
\vspace{-2mm}
\begin{align}
    \frac{1}{2}&\log\frac{\sigma^2_X+1/\nu_{\bar{Z}_l}}{\sigma^2_X+1/\nu_{\bar{Y}_k}}
    = h(\frac{1}{\nu_{\bar{Z}_l}}\bar{Z}_l) - h(\frac{1}{\nu_{\bar{Y}_k}}\bar{Y}_k) \nonumber \\
    &\overset{\rm (a)}\le h(\frac{1}{\nu_{\bar{Z}_l}}\bar{Z}_l|U) - h(\frac{1}{\nu_{\bar{Y}_k}}\bar{Y}_k|U) \nonumber \\
    &\overset{\rm (b)}\le h(\frac{1}{\nu_{\bar{Z}_l}}\bar{Z}_l|X) - h(\frac{1}{\nu_{\bar{Y}_k}}\bar{Y}_k|X) = \frac{1}{2}\log\frac{1/\nu_{\bar{Z}_l}}{1/\nu_{\bar{Y}_k}},
\end{align}
where (a) and (b) follow from the fact that $I(\bar{Y}_k;U|\bar{Z}_l) \ge 0$ and $I(\bar{Y}_k;X|U,\bar{Z}_l) \ge 0$, respectively. Thus, there must exist a parameter $\alpha \in (0,1]$ such that
\begin{align}
 h(\frac{1}{\nu_{\bar{Z}_l}}\bar{Z}_l|U) - h(\frac{1}{\nu_{\bar{Y}_k}}\bar{Y}_k|U) =\frac{1}{2}\log\frac{\alpha\sigma^2_X+1/\nu_{\bar{Z}_l}}{\alpha\sigma^2_X+1/\nu_{\bar{Y}_k}}. \label{izu-iyu}
\end{align}
Equation \eqref{izu-iyu} also indicates that
\begin{align}
h(\bar{Z}_l|U) - h(\bar{Y}_k|U) &=\frac{1}{2}\log\frac{\nu_{\bar{Z}_l}(\alpha\sigma^2_X\nu_{\bar{Z}_l}+1)}{\nu_{\bar{Y}_k}(\alpha\sigma^2_X\nu_{\bar{Y}_k}+1)} \nonumber \\
&=\frac{1}{2}\log\frac{\tH^\intercal_l\tH_l(\alpha\sigma^2_X\tH^\intercal_l\tH_l+1)}{\H^\intercal_k\H_k(\alpha\sigma^2_X\H^\intercal_k\H_k+1)}.
\label{63-63-eq}
\end{align}

The conditional Fisher information of $A$ is defined by
$
\mathbb{J}(A|U) = \mathbb{E}\left[\left(\frac{\partial\log f_{A|U}(a|u)}{\partial a}\right)^2\right],
$
where the expectation is taken over $(U,A)$ \cite[Def. 1]{ekrem-2014}.

\begin{Lemma}[\hspace{-0.1mm}{\cite[Cor. 1]{ekrem-2014}}] \label{key-lemma-2} Let $W, A, B$ be random variables, and let the density for any combination of them exist. Moreover,  assume that given $W$, $A$ and $B$ are independent. Then, we have
\begin{align}
   \frac{1}{\mathbb{J}(A + B|W)} \ge \frac{1}{\mathbb{J}(A|W)} + \frac{1}{\mathbb{J}(B|W)}.
\end{align}
\end{Lemma}

We use Lemma \ref{key-lemma-2} to establish a lower bound on the conditional Fisher information, as presented in Lemma \ref{lower-lemma}. This lemma is then used to derive a lower bound on the difference $h(\bar{Y}_k|U) - h(X|U)$ given that $h(\bar{Z}_l|U) - h(\bar{Y}_k|U)$ is fixed.

\begin{Lemma} \label{lower-lemma} For $0 \le r \le 1/\nu_{\bar{Y}_k}$, it holds that
\begin{align}
    \mathbb{J}(X + \sqrt{r}N|U) \ge \frac{1}{\alpha\sigma^2_X + r} \label{lower-eq}
\end{align}
with an independent Gaussian random variable $N \sim \mathcal{N}(0,1)$.
\end{Lemma}
\begin{proof}
From \cite[Lemma 3]{ekrem-2014}, it follows that
\begin{align}
h(\frac{1}{\nu_{\bar{Z}_l}}\bar{Z}_l|U) - h(\frac{1}{\nu_{\bar{Y}_k}}\bar{Y}_k|U) = \frac{1}{2}\int_{1/\nu_{\bar{Y}_k}}^{1/\nu_{\bar{Z}_l}}\mathbb{J}(X + \sqrt{t}\Tilde{N}|U)dt
\end{align}
with an independent random variable $\Tilde{N} \sim \mathcal{N}(0,1)$. Then,
\begin{align}
    \frac{1}{2}&\int_{1/\nu_{\bar{Y}_k}}^{1/\nu_{\bar{Z}_l}}\mathbb{J}(X + \sqrt{t}\Tilde{N}|U)dt \nonumber \\
    &\overset{\rm (a)}=\frac{1}{2}\int_{1/\nu_{\bar{Y}_k}}^{1/\nu_{\bar{Z}_l}}\mathbb{J}(X + \sqrt{r}N + \sqrt{t-r}N'|U)dt \nonumber \\
    &\overset{\rm (b)}\le \frac{1}{2}\int_{1/\nu_{\bar{Y}_k}}^{1/\nu_{\bar{Z}_l}}\left(\mathbb{J}(X + \sqrt{r}N|U)^{-1} + t-r\right)^{-1}dt \nonumber \displaybreak[0] \\
    &= \frac{1}{2}\int_{1/\nu_{\bar{Y}_k}}^{1/\nu_{\bar{Z}_l}}\frac{\mathbb{J}(X + \sqrt{r}N|U)}{1 + \mathbb{J}(X + \sqrt{r}N|U)(t-r)}dt \nonumber \displaybreak[0] \\
    &= \frac{1}{2}\log\frac{1 + \mathbb{J}(X + \sqrt{r}N|U)( 1/\nu_{\bar{Z}_l}-r)}{1 + \mathbb{J}(X + \sqrt{r}N|U)(1/\nu_{\bar{Y}_k}-r)},
    \label{56-564}
\end{align}
where (a) follows by picking a real number $r$ in the range of $0 \le r \le 1/\nu_{\bar{Y}_k}$ and using independent Gaussian random variables $N \sim \mathcal{N}(0,1)$ and $N'\sim \mathcal{N}(0,1)$, and (b) is due to Lemma \ref{key-lemma-2}. Lastly, comparing \eqref{izu-iyu} and \eqref{56-564}, we obtain~\eqref{lower-eq}.
\end{proof}

Observe that
\begin{align}
h&(\frac{1}{\nu_{\bar{Y}_k}}\bar{Y}_k|U) - h(X|U)= \frac{1}{2}\int_{0}^{1/\nu_{\bar{Y}_k}}\mathbb{J}(X + \sqrt{r}N|U)dr \nonumber \displaybreak[0] \\
&\overset{\rm (a)}\ge \frac{1}{2}\int_{0}^{1/\nu_{\bar{Y}_k}}\frac{1}{\alpha\sigma^2_X + r}dr =\frac{1}{2}\log\frac{\alpha\sigma^2_X + 1/\nu_{\bar{Y}_k}}{\alpha\sigma^2_X},
\end{align}
where (a) follows from Lemma \ref{lower-lemma}, which implies that 
\begin{align}
    h(\bar{Y}_k|U) - h(X|U) &\ge \frac{1}{2}\log\frac{\nu_{\bar{Y}_k}(\alpha\sigma^2_X\nu_{\bar{Y}_k} + 1)}{\alpha\sigma^2_X} \nonumber \\
    &=\frac{1}{2}\log\frac{\H^\intercal_k\H_k(\alpha\sigma^2_X\H^\intercal_k\H_k + 1)}{\alpha\sigma^2_X}. \label{rj-yuxu}
\end{align}

Substituting \eqref{63-63-eq} and \eqref{rj-yuxu} into the rate constraints in \eqref{label-rs}--\eqref{label-rl}, the outer bound for a fixed pair $(k,l)$ is expressed as
\begin{align}
    \bar{\mathcal{O}}_{G_{kl}}
    &\triangleq \bigcup_{0 < \alpha \le 1}\Big\{(R_S,R_J,R_L)\in \mathbb{R}^3_+:~\nonumber \\
    R_S &\le\frac{1}{2}\log\frac{\sigma^2_X\H^\intercal_k\H_k + 1}{\alpha\sigma^2_X\H^\intercal_k\H_k + 1} - \frac{1}{2}\log\frac{\sigma^2_X\tH^\intercal_l\tH_l + 1}{\alpha\sigma^2_X\tH^\intercal_l\tH_l + 1}, \nonumber \\
    R_J &\ge \frac{1}{2}\log\frac{\alpha\sigma^2_X\H^\intercal_k\H_k + 1}{\alpha(\sigma^2_X\H^\intercal_k\H_k + 1)}, \nonumber \\
    R_L &\ge \frac{1}{2}\log\frac{\alpha\sigma^2_X\H^\intercal_k\H_k + 1}{\alpha(\sigma^2_X\H^\intercal_k\H_k + 1)}\Big\}. \label{kl-outer}
\end{align}

Applying Lemma \ref{lemma-fx} to \eqref{discrete-outer-cs-gauss} and \eqref{kl-outer}, the outer region of the GS model for all possible index pairs $(k,l)$ is
\begin{align}
    \RG
    &\subseteq \bigcup_{0 < \alpha \le 1}\Big\{(R_S,R_J,R_L)\in \mathbb{R}^3_+:~\nonumber \\
    R_S &\le\frac{1}{2}\log\frac{(\sigma^2_X\H^\intercal_{k^*}\H_{k^*} + 1)(\alpha\sigma^2_X\tH^\intercal_{l^*}\tH_{l^*} + 1)}{(\alpha\sigma^2_X\H^\intercal_{k^*}\H_{k^*} + 1)(\sigma^2_X\tH^\intercal_{l^*}\tH_{l^*} + 1)}, \nonumber \\
    R_J &\ge \frac{1}{2}\log\frac{\alpha\sigma^2_X\H^\intercal_{k^*}\H_{k^*} + 1}{\alpha(\sigma^2_X\H^\intercal_{k^*}\H_{k^*} + 1)}, \nonumber \\
    R_L &\ge \frac{1}{2}\log\frac{\alpha\sigma^2_X\H^\intercal_{k^*}\H_{k^*} + 1}{\alpha(\sigma^2_X\H^\intercal_{k^*}\H_{k^*} + 1)}\Big\}.
\end{align}

The outer region of the CS model in \eqref{cs-cp-gauss} can be shown similarly.
\qed

\section{Proof of Lemma \ref{lemma7-outer}} \label{proof-lemma7}
We only prove \eqref{o-g-kl}, the outer bound of the GS model for a given pair $(k,l)$, as the proof of \eqref{inner-cs} follows by a similar manner. Assume that a rate tuple $(R_S,R_J,R_L)$ is achievable with respect to Definition \ref{def1} for every pair $(k,l) \in \mathcal{K}\times\mathcal{L}$.

We begin by establishing the following Markov chains:
\begin{align}
    &(J,S)-X^n-\Y^n_k-\bar{Y}^n_k,~(J,S)-X^n-\bar{Y}^n_k-\Y^n_k, \label{54-54} \\
    &(J,S)-X^n-\Z^n_l-\bar{Z}^n_l,~~(J,S)-X^n-\bar{Z}^n_l-\Z^n_l, \label{55-55} \\
   &(J,S)-X^n-(\Y^n_k,\Z^n_l)-(\bar{Y}^n_k,\bar{Z}^n_l), \label{j-55-55}
\end{align}
where the left-hand sides of \eqref{54-54} and \eqref{55-55}, and \eqref{j-55-55} hold because $\bar{Y}^n_k$ and $\bar{Z}^n_l$ are functions of $\Y^n_k$ and $\Z^n_l$, respectively, by Lemma \ref{bar-relation}, and the right-hand sides of \eqref{54-54} and \eqref{55-55} are due to the sufficient statistic \cite[Sect. 2.9]{cover}. In addition, for the scalar random variables $(X,\bar{Y}_k,\bar{Z}_l)$, the Markov chain $X^n-\bar{Y}^n_k-\bar{Z}^n_l$ holds for any pair $(k,l)$ when $\H^\intercal_{k^*}\H_{k^*} \ge \tH^\intercal_{l^*}\tH_{l^*}$. Combining this with \eqref{j-55-55} gives
\begin{align}
    (J,S)-X^n-\bar{Y}^n_k-\bar{Z}^n_l. \label{long-markov-2}
\end{align}

Define auxiliary random variables
\begin{align}
V_t = (J,\bar{Y}^n_{k,t+1},\bar{Z}^{t-1}_l)~{\rm and}~U_t = (J,S,\bar{Y}^n_{k,t+1},\bar{Z}^{t-1}_l),
\end{align}
which guarantee the Markov chain
\begin{align}
    V_t-U_t-X_t-\bar{Y}_{k,t}-\bar{Z}_{l,t}. \label{60-60}
\end{align}

Also, we define
\begin{align}
    \delta_n = \frac{1}{n}\left(H_b(\delta) + \delta\log|\mathcal{S}|\right), \label{delta-n}
\end{align}
where $H_b(\cdot)$ denotes the binary entropy function, and $\delta_n \downarrow 0$ as $\delta \downarrow 0$ and $n \rightarrow \infty$.

\smallskip
{\em Analysis of Secret-Key Rate}: From \eqref{secretk},
\begin{align}
    &n(R_S - \delta) \le H(S) \nonumber \\
    &= H(S|J,\Z^n_l) + I(S;J,\Z^n_l) \nonumber \\
    &\overset{\rm (a)}\le H(S|J,\Z^n_l) - H(S|J,\Y^n_k) + n(\delta + \delta_n) \nonumber \\
    &\overset{\rm (b)}= H(S|J,\Z^n_l,\bar{Z}^n_l) - H(S|J,\Y^n_k,\bar{Y}^n_k) + n(\delta + \delta_n) \nonumber \displaybreak[0]\\
    &\overset{\rm (c)}= H(S|J,\bar{Z}^n_l) - H(S|J,\bar{Y}^n_k) + n(\delta + \delta_n) \nonumber \displaybreak[0]\\
    &= I(S;\bar{Y}^n_k|J)-I(S;\bar{Z}^n_l|J) + n(\delta + \delta_n) \nonumber \displaybreak[0]\\
    &\overset{\rm (d)} = \sum_{t=1}^n\{I(\bar{Y}_{k,t};U_t|V_t)-I(\bar{Z}_{l,t};U_t|V_t)\} + n(\delta + \delta_n) \label{8080-80} \displaybreak[0]\\
    &\overset{\rm (e)}= \sum_{t=1}^n\{I(\bar{Y}_{k,t};U_t)-I(\bar{Z}_{l,t};U_t) \nonumber \\
    &~~~-(I(\bar{Y}_{k,t};V_t)-I(\bar{Z}_{l,t};V_t))\} + n(\delta + \delta_n) \nonumber \\
    &\overset{\rm (f)}\le \sum_{t=1}^n\{I(\bar{Y}_{k,t};U_t)-I(\bar{Z}_{l,t};U_t)\} + n(\delta + \delta_n), \label{rs-64}
\end{align}
where (a) is due to \eqref{secrecy} and Fano's inequality with $\delta_n$ defined in \eqref{delta-n} as the secret key $S$ can be reliably estimated from $(J,\Y^n_k)$, (b) follows from the left-hand sides of \eqref{54-54} and \eqref{55-55}, (c) holds by the right-hand sides of \eqref{54-54} and \eqref{55-55}, (d) follows by \cite[Lemma 4.1]{ac1993}, (e) follows from the Markov chains $V_t-U_t-\bar{Y}_{k,t}$ and $V_t-U_t-\bar{Z}_{l,t}$, and (f) is due to \eqref{60-60}, which results in $I(\bar{Y}_{k,t};V_t)-I(\bar{Z}_{l,t};V_t) \ge 0$.

\smallskip
{\em Analysis of Storage Rate}: From \eqref{storage},
\begin{align}
    n&(R_J + \delta) \ge \log |\mathcal{J}| \ge H(J) = I(X^n;J) \nonumber \\
    &\overset{\rm (a)}\ge I(X^n;J|\Y^n_k) \nonumber \\
    &\ge I(X^n;J,S|\Y^n_k)-H(S|\Y^n_k,J) \nonumber \displaybreak[0]\\
    &\overset{\rm (b)}\ge I(X^n;J,S|\Y^n_k)-n\delta_n \nonumber \displaybreak[0]\\
    &= I(X^n;\bar{Y}^n_k|\Y^n_k) + I(X^n;J,S|\bar{Y}^n_k,\Y^n) \nonumber \\
    &~~~~- I(X^n;\bar{Y}^n_k|J,S,\Y^n_k)-n\delta_n \nonumber \displaybreak[0]\\
    & \overset{\rm (c)}= I(X^n;J,S|\bar{Y}^n_k,\Y^n_k) -n\delta_n \nonumber \displaybreak[0]\\
    &\overset{\rm (d)} = I(X^n;J,S|\bar{Y}^n_k)-n\delta_n \label{64-64} \displaybreak[0]\\
    &\overset{\rm (e)} = \sum_{t=1}^n\{h(X_t|\bar{Y}_{k,t})-h(X_t|J,S,X^{t-1},\bar{Y}^n_k,\bar{Z}^{t-1}_l)\}-n\delta_n \nonumber \displaybreak[0] \\
    &\overset{\rm (f)} \ge \sum_{t=1}^n\{h(X_t|\bar{Y}_{k,t})-h(X_t|U_t,\bar{Y}_{k,t})\}-n\delta_n \nonumber \displaybreak[0]\\
    &\overset{\rm (g)}= \sum_{t=1}^n\{I(X_t;U_t)-I(\bar{Y}_{k,t};U_t)\}-n\delta_n, \label{65-65}
\end{align}
where (a) is due to the Markov chain $J-X^n-\Y^n_k$, (b) follows by Fano's inequality with $\delta_n$ defined in \eqref{delta-n}, (c) follows because $\bar{Y}^n_k$ is a function of $\Y^n_k$, (d) holds from the left-hand side of \eqref{54-54}, (e) holds due to the Markov chain $X_t - (J,S,X^{t-1},\bar{Y}^n_k,)-\bar{Z}^{t-1}_l$, (f) follows because conditioning reduces entropy, and (g) is due to $U_t-X_t-\bar{Y}_{k,t}$.

\smallskip
{\em Analysis of Privacy-Leakage Rate}: For a fixed $l$, we first show that the left-hand side of \eqref{privacy} is preserved when $(X^n,\Z^n)$ is replaced with $(X^n,\bar{Z}^n)$.
\begin{align}
    I(X^n;J|\Z^n_l) &\overset{\rm (a)}= I(X^n;J) - I(\Z^n_l;J) \nonumber \\
    &\overset{\rm (b)}= I(X^n;J) - I(\Z^n_l,\bar{Z}^n_l;J) \nonumber \\
    &\overset{\rm (c)}= I(X^n;J) - I(\bar{Z}^n_l;J) \nonumber \\
    &\overset{\rm (d)}= I(X^n;J|\bar{Z}^n_l), \label{preserve-61}
\end{align}
where (a), (b), (c), and (d) follow from the Markov chains $J-X^n-\Z^n_l$, $J-\Z^n_l-\bar{Z}^n_l$, $J-\bar{Z}^n_l-\Z^n_l$, and $J-X^n_l-\bar{Z}^n_l$, respectively, all of which are obtained as special cases of \eqref{55-55}.

Therefore, we can evaluate the privacy-leakage rate as
\begin{align}
    &n(R_L + \delta) \ge I(X^n;J|\Z^n_l) \nonumber \\ 
    &\overset{\rm (a)}=I(X^n;J|\bar{Z}^n_l) \nonumber \displaybreak[0]\\
    &=I(X^n;J,S,\Y^n_k|\bar{Z}^n_l) - I(X^n;\Y^n_k|J,\bar{Z}^n_l) \nonumber\displaybreak[0] \\
    &~~- I(X^n;S|J,\Y^n_k,\bar{Z}^n_l) \nonumber \displaybreak[0]\\
    &\ge I(X^n;J,S,\Y^n_k|\bar{Z}^n_l) - I(X^n;\Y^n_k|J,\bar{Z}^n_l) - H(S|J,\Y^n_k) \nonumber \displaybreak[0] \\
    &\ge I(X^n;J,S,\Y^n_k|\bar{Z}^n_l) - I(X^n;\Y^n_k|J,\bar{Z}^n_l) -n\delta_n \nonumber \\
    &\overset{\rm (b)}= I(X^n;J,S,\bar{Y}^n_k,\Y^n_k|\bar{Z}^n_l) - I(X^n;\bar{Y}^n_k,\Y^n_k|J,\bar{Z}^n_l) -n\delta_n \nonumber \\
    &\overset{\rm (c)}= I(X^n;J,S,\bar{Y}^n_k|\bar{Z}^n_l) - I(X^n;\bar{Y}^n_k|J,\bar{Z}^n_l) -n\delta_n \nonumber \displaybreak[0]\\
    &\overset{\rm (d)}= I(X^n;J,S|\bar{Y}^n_k,\bar{Z}^n_l) + I(X^n;\bar{Y}^n_k|\bar{Z}^n_l) \nonumber \\
    &~~~- (I(X^n;\bar{Y}^n_k|\bar{Z}^n_l) - I(\bar{Y}^n_k;J|\bar{Z}^n_l))  -n\delta_n \nonumber \displaybreak[0]\\
    &\ge I(X^n;J,S|\bar{Y}^n_k,\bar{Z}^n_l) -n\delta_n \nonumber \displaybreak[0]\\
    &\overset{\rm (e)}= I(X^n;J,S|\bar{Y}^n_k)  -n\delta_n \nonumber \displaybreak[0]\\
    &\overset{\rm (f)}\ge \sum_{t=1}^n\{I(X_t;U_t)-I(\bar{Y}_{k,t};U_t)\}-n\delta_n, \label{rl-68}
\end{align}
where (a) is due to \eqref{preserve-61}, (b) follows from \eqref{barxyz}, i.e., $\bar{Y}^n_k$ is a function of $\Y^n_k$, (c) holds because the Markov chain $X^n-(J,S,\bar{Y}^n_k,\bar{Z}^n_l)-\Y^n_k$ holds, from the right-hand side of \eqref{54-54}, (d) and (e) are due to the Markov chain $(J,S)-X^n-\bar{Y}^n_k-\bar{Z}^n_l$, obtained from \eqref{long-markov-2}, and (e) follows by the same steps from \eqref{64-64} to \eqref{65-65}.

For the case where $\H^\intercal_{k^*}\H_{k^*} < \tH^\intercal_{l^*}\tH_{l^*}$, the Markov chain $V_t-U_t-X_t-\bar{Z}_{l,t}-\bar{Y}_{k,t}$ holds. The secret-key rate follows from \eqref{8080-80} since $I(\bar{Y}_{k,t};U_t|V_t) \le I(\bar{Z}_{l,t};U_t|V_t)$.

Finally, we introduce a time-sharing random variable $Q \sim {\rm Unif}[1:n]$, independent of other random variables, and define $U = (U_Q,Q)$, $X = X_Q$, $\bar{Y}_k= \bar{Y}_{k,Q}$, and $\bar{Z}_l = \bar{Z}_{l,Q}$, so that the Markov chain $U-X-\bar{Y}_k-\bar{Z}_l$ holds. By letting $n \rightarrow \infty$ and $\delta \downarrow 0$, one can see that for a given pair $(k,l)$, the outer bound of the GS model is given by \eqref{o-g-kl}. Hence, the outer bound valid for any pair $(k,l)$ is given by \eqref{discrete-outer-cs-gauss}.
\qed

\bibliographystyle{IEEEtran}

\begin{thebibliography}{10}
\providecommand{\url}[1]{#1}
\csname url@samestyle\endcsname
\providecommand{\newblock}{\relax}
\providecommand{\bibinfo}[2]{#2}
\providecommand{\BIBentrySTDinterwordspacing}{\spaceskip=0pt\relax}
\providecommand{\BIBentryALTinterwordstretchfactor}{4}
\providecommand{\BIBentryALTinterwordspacing}{\spaceskip=\fontdimen2\font plus
\BIBentryALTinterwordstretchfactor\fontdimen3\font minus \fontdimen4\font\relax}
\providecommand{\BIBforeignlanguage}[2]{{%
\expandafter\ifx\csname l@#1\endcsname\relax
\typeout{** WARNING: IEEEtran.bst: No hyphenation pattern has been}%
\typeout{** loaded for the language `#1'. Using the pattern for}%
\typeout{** the default language instead.}%
\else
\language=\csname l@#1\endcsname
\fi
#2}}
\providecommand{\BIBdecl}{\relax}
\BIBdecl

\bibitem{va-remi-itw}
V.~Yachongka and R.~Chou, ``Secret-key generation with {PUF}s and biometric identifiers for compound authentication channels,'' in \emph{Proc. Inf. Theory Workshop (ITW)}, Shenzhen, China, Nov. 2024, pp. 591--596.

\bibitem{Mohamad-2022}
M.~Ebrahimabadi, M.~Younis, and N.~Karimi, ``A {PUF}-based modeling-attack resilient authentication protocol for {I}o{T} devices,'' \emph{IEEE Internet Things J.}, vol.~9, no.~5, pp. 3684--3703, Mar. 2022.

\bibitem{gunlue2020}
O.~Günlü and R.~F. Schaefer, ``An optimality summary: Secret key agreement with physical unclonable functions,'' \emph{Entropy}, vol.~23, no.~1, Jan. 2021.

\bibitem{pappu2001}
R.~Pappu, B.~Recht, J.~Taylor, and N.~Gershenfeld, ``Physical one-way functions,'' \emph{Science}, vol. 297, no. 5589, pp. 2026--2030, Sep. 2002.

\bibitem{puf2002}
B.~Gassend, D.~Clarke, M.~van Dijk, and S.~Devadas, ``Silicon physical random functions,'' in \emph{Proc. the 9th ACM Conf. Comput. Commun. Secur. (CCS)}, Washington DC, USA, Nov. 2002, pp. 148--160.

\bibitem{suh-2007}
G.~E. Suh and S.~Devadas, ``Physical unclonable functions for device authentication and secret key generation,'' in \emph{44th ACM/IEEE Design Autom. Conf. (DAC)}, San Diego, CA, USA, Jun. 2007, pp. 9--14.

\bibitem{ac1993}
R.~Ahlswede and I.~Csiszar, ``Common randomness in information theory and cryptography. {I}. {S}ecret sharing,'' \emph{IEEE Trans. Inf. Theory}, vol.~39, no.~4, pp. 1121--1132, Jul. 1993.

\bibitem{gebali-22}
F.~Gebali and M.~Mamun, ``Review of physically unclonable functions ({PUF}s): Structures, models, and algorithms,'' \emph{Front. Sens.}, vol.~2, Jan. 2022.

\bibitem{itw3}
T.~Ignatenko and F.~M.~J. Willems, ``Biometric systems: Privacy and secrecy aspects,'' \emph{IEEE Trans. Inf. {Forensics} Secur.}, vol.~4, no.~4, pp. 956--973, Dec. 2009.

\bibitem{lhp}
L.~Lai, S.-W. Ho, and H.~V. Poor, ``Privacy-security trade-offs in biometric security systems--part {I}: Single use case,'' \emph{IEEE Trans. Inf. {Forensics} Secur.}, vol.~6, no.~1, pp. 122--139, Mar. 2011.

\bibitem{iw-book}
T.~Ignatenko and F.~M. Willems, ``Biometric security from an information-theoretical perspective,'' \emph{Found. Trends Commun. Inform. Theory}, vol.~7, no. 2–3, pp. 135--316, Feb. 2012.

\bibitem{chou2019}
R.~A. Chou, ``Biometric systems with multiuser access structures,'' in \emph{Proc. IEEE Int. Symp. Inf. Theory (ISIT)}, Paris, France, Jul. 2019, pp. 807--811.

\bibitem{kuster2019}
L.~Kusters and F.~M.~J. Willems, ``Secret-key capacity regions for multiple enrollments with an {SRAM-PUF},'' \emph{IEEE Trans. Inf. {Forensics} Secur.}, vol.~14, no.~9, pp. 2276--2287, Jan. 2019.

\bibitem{KY}
M.~Koide and H.~Yamamoto, ``Coding theorems for biometric systems,'' in \emph{Proc. IEEE Int. Symp. Inf. Theory (ISIT)}, Austin, TX, USA, Jun. 2010.

\bibitem{kc2015}
K.~Kittichokechai and G.~Caire, ``Secret key-based authentication with a privacy constraint,'' in \emph{Proc. IEEE Int. Symp. Inf. Theory (ISIT)}, Hong Kong, China, Jun. 2015, pp. 1791--1795.

\bibitem{vyo2022-isit}
V.~Yachongka, H.~Yagi, and Y.~Oohama, ``Secret key-based authentication with passive eavesdropper for scalar {G}aussian sources,'' in \emph{Proc. IEEE Int. Symp. Inf. Theory (ISIT)}, Espoo, Finland, Jun./Jul. 2022, pp. 2685--2690.

\bibitem{cn2000}
I.~Csiszar and P.~Narayan, ``Common randomness and secret key generation with a helper,'' \emph{IEEE Trans. Inf. Theory}, vol.~46, no.~2, pp. 344--366, Mar. 2000.

\bibitem{wataoha-tifs}
S.~Watanabe and Y.~Oohama, ``Secret key agreement from vector {G}aussian sources by rate limited public communication,'' \emph{IEEE Trans. Inf. Forensics Secur.}, vol.~6, no.~3, pp. 541--550, Sep. 2011.

\bibitem{chou2015}
R.~A. Chou, M.~R. Bloch, and E.~Abbe, ``Polar coding for secret-key generation,'' \emph{IEEE Trans. Inf. Theory}, vol.~61, no.~11, pp. 6213--6237, Nov. 2015.

\bibitem{gunlu2018}
O.~G\"{u}nl\"{u} and G.~Kramer, ``Privacy, secrecy, and storage with multiple noisy measurements of identifiers,'' \emph{IEEE Trans. Inf. {Forensics} Secur.}, vol.~13, no.~11, pp. 2872--2883, Nov. 2018.

\bibitem{gksc2018}
O.~G\"{u}nl\"{u}, K.~Kittichokechai, R.~F. Schaefer, and G.~Caire, ``Controllable identifier measurements for private authentication with secret keys,'' \emph{IEEE Trans. Inf. {Forensics} Secur.}, vol.~13, no.~8, pp. 1945--1959, Aug. 2018.

\bibitem{vamoua-tifs}
V.~Yachongka, H.~Yagi, and H.~Ochiai, ``Key agreement using physical identifiers for degraded and less noisy authentication channels,'' \emph{IEEE Trans. Inf. Forensics Secur.}, vol.~18, pp. 5316--5331, Aug. 2023.

\bibitem{itw-tit-2015}
T.~Ignatenko and F.~M.~J. Willems, ``Fundamental limits for privacy-preserving biometric identification systems that support authentication,'' \emph{IEEE Trans. Inf. Theory}, vol.~61, no.~10, pp. 5583--5594, Oct. 2015.

\bibitem{kitti-2016}
K.~Kittichokechai and G.~Caire, ``Secret key-based identification and authentication with a privacy constraint,'' \emph{IEEE Trans. Inf. Theory}, vol.~62, no.~11, pp. 6189--6203, Nov. 2016.

\bibitem{vy2020}
V.~Yachongka and H.~Yagi, ``A new characterization of the capacity region of identification systems under noisy enrollment,'' in \emph{Proc. 54th Annu. Conf. Inf. Sci. Syst. (CISS)}, Princeton, NJ, USA, Mar. 2020, pp. 1--6.

\bibitem{zhou2022}
L.~Zhou, T.~J. Oechtering, and M.~Skoglund, ``Fundamental limits-achieving polar code designs for biometric identification and authentication,'' \emph{IEEE Trans. Inf. {Forensics} Secur.}, vol.~17, pp. 180--195, Dec. 2022.

\bibitem{tavangaran2017}
N.~Tavangaran, S.~Baur, A.~Grigorescu, and H.~Boche, ``Compound biometric authentication systems with strong secrecy,'' in \emph{Proc. 11th Int. ITG Conf. Syst., Commun. Coding (SCC)}, Hamburg, Germany, Feb. 2017, pp. 1--5.

\bibitem{Grigorescu2017}
A.~Grigorescu, H.~Boche, and R.~F. Schaefer, ``Robust biometric authentication from an information theoretic perspective,'' \emph{Entropy}, vol.~19, no.~9, Sep. 2017.

\bibitem{colombier-2017}
B.~Colombier, L.~Bossuet, V.~Fischer, and D.~Hély, ``Key reconciliation protocols for error correction of silicon {PUF} responses,'' \emph{IEEE Trans. Inf. Forensics Secur.}, vol.~12, no.~8, pp. 1988--2002, Aug. 2017.

\bibitem{ANA2021}
N.~N. Anandakumar, M.~S. Hashmi, and M.~Tehranipoor, ``{FPGA}-based physical unclonable functions: A comprehensive overview of theory and architectures,'' \emph{Integration}, vol.~81, pp. 175--194, Nov. 2021.

\bibitem{Jain-2008}
A.~K. Jain, K.~Nandakumar, and A.~Nagar, ``Biometric template security,'' \emph{EURASIP J. Adv. Signal Process.}, vol. 2008, no.~1, pp. 1--7, Apr. 2008.

\bibitem{Ignatenko-2010}
T.~Ignatenko and F.~M.~J. Willems, ``Information leakage in fuzzy commitment schemes,'' \emph{IEEE Trans. Inf. Forensics Secur.}, vol.~5, no.~2, pp. 337--348, 2010.

\bibitem{gunlu-icc}
O.~Günlü, A.~Belkacem, and B.~C. Geiger, ``Secret-key binding to physical identifiers with reliability guarantees,'' in \emph{Proc. IEEE Int. Conf. Commun. (ICC)}, Paris, France, May 2017, pp. 1--6.

\bibitem{vu-2021}
M.~T. Vu, T.~J. Oechtering, M.~Skoglund, and H.~Boche, ``Uncertainty in identification systems,'' \emph{IEEE Trans. Inf. Theory}, vol.~67, no.~3, pp. 1400--1414, Mar. 2021.

\bibitem{zhou-allerton}
L.~Zhou, T.~J. Oechtering, and M.~Skoglund, ``Uncertainty in biometric identification and authentication systems with strong secrecy,'' in \emph{Proc. 58th Annu. Allert. Conf. Commun. Control Comput. Allert. (Allerton)}, Monticello, IL, USA, Sep. 2022, pp. 1--6.

\bibitem{Remi-2013-GlobalSIP}
R.~Chou and M.~R. Bloch, ``Secret-key generation with arbitrarily varying eavesdropper's channel,'' in \emph{Proc. IEEE Global Conf. Signal Inf. Process. (GlobalSIP)}, Austin, TX, USA, Dec. 2013, pp. 277--280.

\bibitem{tavangaran2017-tifs}
N.~Tavangaran, H.~Boche, and R.~F. Schaefer, ``Secret-key generation using compound sources and one-way public communication,'' \emph{IEEE Trans. Inf. Forensics Secur.}, vol.~12, no.~1, pp. 227--241, Jan. 2017.

\bibitem{tavangaran-2018}
N.~Tavangaran, R.~F. Schaefer, H.~V. Poor, and H.~Boche, ``Secret-key generation and convexity of the rate region using infinite compound sources,'' \emph{IEEE Trans. Inf. Forensics Secur.}, vol.~13, no.~8, pp. 2075--2086, Feb. 2018.

\bibitem{sultana-2021}
R.~Sultana and R.~A. Chou, ``Low-complexity secret sharing schemes using correlated random variables and rate-limited public communication,'' in \emph{Proc. IEEE Int. Symp. Inf. Theory (ISIT)}, Melbourne, Australia, Jul. 2021, pp. 970--975.

\bibitem{liang-compound}
Y.~Liang, G.~Kramer, H.~V. Poor, and S.~Shamai~(Shitz), ``Compound wiretap channels,'' \emph{EURASIP J. Wirel. Commun. Netw.}, vol. 2009, no.~1, pp. 1--12, Mar. 2009.

\bibitem{Bjelakovic-2013}
I.~Bjelakovi\'c, H.~Boche, and J.~Sommerfeld, ``Secrecy results for compound wiretap channel,'' \emph{{Probl. Inf. Transm.}}, vol.~49, no.~1, p. 73–98, Apr. 2013.

\bibitem{campello-2020}
A.~Campello, C.~Ling, and J.-C. Belfiore, ``Semantically secure lattice codes for compound {MIMO} channels,'' \emph{IEEE Trans. Inf. Theory}, vol.~66, no.~3, pp. 1572--1584, Mar. 2020.

\bibitem{remi-2022-tifs}
R.~A. Chou, ``Explicit wiretap channel codes via source coding, universal hashing, and distribution approximation, when the channels’ statistics are uncertain,'' \emph{IEEE Trans. Inf. Forensics Secur.}, vol.~18, pp. 117--132, Oct. 2023.

\bibitem{CK-book}
I.~Csisz\'{a}r and J.~K\"{o}rner, \emph{Information Theory: Coding Theorems for Discrete Memoryless Systems}.\hskip 1em plus 0.5em minus 0.4em\relax Cambridge, U. K.: Cambridge Univ. Press, 2011.

\bibitem{weingarten-2006}
H.~Weingarten, Y.~Steinberg, and S.~Shamai, ``The capacity region of the gaussian multiple-input multiple-output broadcast channel,'' \emph{IEEE Trans. Inf. Theory}, vol.~52, no.~9, pp. 3936--3964, Sep. 2006.

\bibitem{cover}
T.~M. Cover and J.~A. Thomas, \emph{Elements of Information Theory}.\hskip 1em plus 0.5em minus 0.4em\relax NJ, USA: 2nd ed. John Wiley \& Sons, 2006.

\bibitem{wataoha2010}
S.~Watanabe and Y.~Oohama, ``Secret key agreement from correlated {G}aussian sources by rate limited public communication,'' \emph{IEICE Trans. Fundam. Electron., Commun. Comput. Sci.}, vol. E93-A, no.~11, pp. 1976--1983, Nov. 2010.

\bibitem{vy2}
V.~Yachongka, H.~Yagi, and Y.~Oohama, ``Biometric identification systems with noisy enrollment for {G}aussian sources and channels,'' \emph{Entropy}, vol.~23, no.~8, Aug. 2021.

\bibitem{ekrem-2014}
E.~Ekrem and S.~Ulukus, ``The secrecy capacity region of the {G}aussian {MIMO} multi-receiver wiretap channel,'' \emph{IEEE Trans. Inf. Theory}, vol.~57, no.~4, pp. 2083--2114, Apr. 2011.

\bibitem{kelkboom-2010}
E.~J.~C. Kelkboom, G.~Garcia~Molina, J.~Breebaart, R.~N.~J. Veldhuis, T.~A.~M. Kevenaar, and W.~Jonker, ``Binary biometrics: An analytic framework to estimate the performance curves under gaussian assumption,'' \emph{IEEE Trans. Syst. Man Cybern. A Syst. Hum.}, vol.~40, no.~3, pp. 555--571, May 2010.

\bibitem{GK}
A.~El~Gamal and Y.-H. Kim, \emph{Network Information Theory}.\hskip 1em plus 0.5em minus 0.4em\relax Cambridge, U. K.: Cambridge Univ. Press, 2011.

\bibitem{Pozrikidis2014}
C.~Pozrikidis, \emph{An Introduction to Grids, Graphs, and Networks}.\hskip 1em plus 0.5em minus 0.4em\relax New York, NY, USA: Oxford Univ. Press, 2014.

\bibitem{remi-2014}
R.~A. Chou and M.~R. Bloch, ``Separation of reliability and secrecy in rate-limited secret-key generation,'' \emph{IEEE Trans. Inf. Theory}, vol.~60, no.~8, pp. 4941--4957, Aug. 2014.

\bibitem{gallager-2013}
R.~G. Gallager, \emph{Stochastic Processes: Theory for Applications}.\hskip 1em plus 0.5em minus 0.4em\relax Cambridge, U. K.: Cambridge Univ. Press, 2013.

\bibitem{rana2022}
V.~Rana, R.~Chou, and H.~M. Kwon, ``Information-theoretic secret sharing from correlated {G}aussian random variables and public communication,'' \emph{IEEE Trans. Inf. Theory}, vol.~68, no.~1, pp. 549--559, Jan. 2022.

\bibitem{parada-isit2005}
P.~Parada and R.~Blahut, ``Secrecy capacity of {SIMO} and slow fading channels,'' in \emph{Proc. IEEE Int. Symp. Inf. Theory (ISIT)}, Adelaide, SA, Australia, Sep. 2005, pp. 2152--2155.

\end{thebibliography}

\end{document}